\documentclass[11pt]{article}
\usepackage{amsmath, amssymb}
\usepackage{iclr2026_conference,times}
\usepackage{todonotes}
\usepackage{hyperref}
\usepackage{cleveref}
\usepackage{booktabs}
\usepackage{algorithm}
\usepackage{algpseudocode}
\usepackage{xcolor}
\usepackage{comment}
\usepackage{graphicx}
\usepackage{nicefrac}
\usepackage{subcaption}
\usepackage{macros}
\usepackage{booktabs}
\usepackage{siunitx}
\usepackage{authblk}

\usepackage{amsmath,amsfonts,bm}

\def\eqref#1{equation~\ref{#1}}

\def\1{\bm{1}}

\DeclareMathAlphabet{\mathsfit}{\encodingdefault}{\sfdefault}{m}{sl}
\SetMathAlphabet{\mathsfit}{bold}{\encodingdefault}{\sfdefault}{bx}{n}

\newcommand{\E}{\mathbb{E}}

\newcommand{\R}{\mathbb{R}}

\DeclareMathOperator*{\argmax}{arg\,max}

\newcommand{\emdash}{\,---\,}
\newcommand{\hatd}{\tilde{\mathbf{d}}}
\newcommand{\hatp}{\hat{\mathbf{p}}}

\newcommand{\inst}{\mathcal{I}}

\newcommand{\FV}{\mathrm{FV}}
\DeclareMathOperator{\dev}{dev}

\title{Near-Optimal Dropout-Robust Sortition}

\author{
~~Maya Pal Gambhir$^{1}$,~~Bailey Flanigan$^{2}$,~~Aaron Roth$^1$\\ 
$^1$University of Pennsylvania, $^2$MIT\\ 
}

\begin{document}
\iclrfinalcopy
\maketitle

\begin{abstract}
Citizens' assemblies\,--\,small panels of citizens that convene to deliberate on policy issues\,--\,often face the issue of panelists dropping out at the last-minute. Without intervention, these dropouts compromise the size and representativeness of the panel, prompting the question: Without seeing the dropouts ahead of time, can we choose panelists such that \textit{after} dropouts, the panel will be representative and appropriately-sized? We model this problem as a minimax game: the minimizer aims to choose a panel that minimizes the \textit{loss}, i.e., the deviation of the ultimate panel from predefined representation targets. Then, an adversary defines a distribution over dropouts from which the realized dropouts are drawn. Our main contribution is an efficient loss-minimizing algorithm, which remains optimal as we vary the maximizer's power from worst case to average case. Our algorithm\,--\,which iteratively plays a projected gradient descent subroutine against an efficient algorithm for computing the best-response dropout distribution\,--\,also addresses a key open question in the area: how to manage dropouts while ensuring that each potential panelist is chosen with relatively \textit{equal} probabilities. Using real-world datasets, we compare our algorithms to existing benchmarks, and we offer the first characterizations of tradeoffs between robustness, loss, and equality in this problem.

\end{abstract}

\section{Introduction} \label{sec:intro}

\textit{Citizens assemblies} are a method for democratic decision-making where a representative group of citizens is randomly selected to convene, deliberate, and produce collective policy proposals. Citizens' assemblies are on the rise globally: in just the past few years, assemblies have been run by, e.g., the European Commission on mobility and food waste \citep{ec_learning_mobility_2024,ec_food_waste_panel_2023}; the Netherlands on climate policy \citep{nl_burgerberaad_klimaat_2025}; and Ireland on drug use \citep{ireland_cadu_report_2024}.

We study the task of \textit{sortition}, i.e., the random selection of the \textbf{\textit{panel}} of assembly participants. In sortition, one must randomly choose a panel from a given \textbf{\textit{pool}} of volunteers. This pool is usually highly self-selected on important demographic and ideological dimensions. Motivated by the demands of sortition in practice and decades of political theory (e.g., see \cite{carson1999random}), we adopt the same three algorithmic ideals for the sortition process as considered in past work (see \Cref{sec:relatedwork}):

\textsc{I)  Panel size.} The panel should be of some fixed, practitioner-chosen size, due to logistical constraints and attributed costs per participant. Let this desired panel size be $k \in \mathbb{N}$.

\textsc{II)  Representation.} The panel must be ``representative'' of the population. Although the precise definition of representation is debated, the ability for practitioners to actively \textit{control} representation is essential, to ensure they do not replicate the skew of the pool. In practice, this control is usually achieved via hard \textit{quotas}: practitioner-defined upper and lower bounds on the number of panelists from predefined groups (defined by, e.g., age, education, political leaning). 

\textsc{III)  Equality.} Finally, the sortition process should select all pool members with relatively equal \textit{selection probabilities}, i.e., chances of being chosen for the panel. While the need to reverse the skew of the pool makes it mathematically impossible to give pool members \textit{exactly equal} selection probabilities, making these selection probabilities \textit{as equal as possible} is a core normative principle of sortition \citep{carson1999random}. Further, past work formally shows that very \textit{high} or very \textit{low} selection probabilities concretely compromise fairness and create manipulation incentives \citep{baharav2024fair}. We therefore adopt the perspective that a lottery is more \textit{equal} if its maximum selection probability is lower and its minimum is higher; formally, for $\alpha,\beta \in [0,1]$ with $\alpha \leq \beta$, we say that a sortition procedure is $(\alpha,\beta)$-equal if the selection probabilities it realizes are bounded in $[\alpha,\beta]$. 

Given ideals \textsc{I-III}, the \textbf{standard sortition task} is then: given a pool, sample a panel from that pool that satisfies desiderata \textsc{I} and \textsc{II}, while achieving desideratum \textsc{III} to the highest degree possible. 

\textbf{The dropouts problem.} After much research on the standard sortition task, a recent paper by \citet{assos2025alternates} pointed out a threat to \textsc{I. panel size} and \textsc{II. representation} that was not encompassed by this task's standard formulation: \textbf{\textit{dropouts}}. In practice, some panelists often drop out on or around the first day of the assembly, compromising \textsc{representation} and/or leaving the \textsc{panel size} too small. Assos et al. identified various measures that practitioners take in anticipation of dropouts: choosing \textit{alternates} (people to have on hand to replace dropouts as needed), or choosing a larger panel initially to make it \textit{robust} to dropouts. With either approach, the key challenge is that dropouts are last-minute, so the alternates/panel must be selected \textit{before} seeing who drops out.

\cite{assos2025alternates} modeled each panelist's dropout event as an independent Bernoulli, and they proposed an empirical risk minimization (ERM)-based algorithm (henceforth called \textit{ERM-Benchmark}) that does the following: given the probabilities that each pool member will drop out if chosen for the panel, find the optimal set of alternates / robust panel. Here, ``optimal'' means the set that minimizes the \textit{Loss}, i.e., the deviation of the final panel from the quotas, in expectation over the randomly-drawn dropouts. 

\textbf{Our goal.} Our high-level goal is to design a sortition procedure that directly produces a dropout-robust panel (rather than selecting alternates).\footnote{Selecting a dropout-robust panel directly (rather than first selecting a panel, \textit{then} choosing extras or alternates) is logistically simpler, and allows us to more directly examine fundamental trade-offs without the additional constraint that a panel has been pre-chosen by an external algorithm.} We want to satisfy ideals \textsc{I - III} \textit{accounting for dropouts:} This means that the panel \textit{after dropouts and corrective measures} would ideally be of size $k$ (\textsc{I}) and satisfy the quotas (\textsc{II}). For ideal \textsc{III}, we want our procedure to ensure that no pool member has more than $\beta$ or less than $\alpha$ chance of reaching the panel.

Our proposed sortition algorithm aims to close two remaining gaps in how the \textit{ERM-Benchmark} addresses these ideals, outlined below.

\textbf{Gap 1.} \textit{Robustness of \textsc{I. Panel Size}, \textsc{II. Representation} to correlated dropout events with unknown probabilities.} The \textit{ERM-Benchmark} takes the dropout probabilities as inputs, but in practice they are actually unknown. While Assos et al. points out that they can be estimated from past assembly data, such data is limited by difficulties combining data containing heterogeneous features, and dropout probabilities can be unpredictably affected by assembly-specific shocks (e.g., the policy topic, communication failures, weather). Put more formally, we contend that if $\dvec \in [0,1]^n$ is the \textit{true} vector of marginal probabilities with which each pool member in $1 \dots n$ will drop out, our estimated version of this vector $\hatd \in [0,1]^n$ could be quite different. \citet{assos2025alternates} do tightly characterize their algorithms' loss under prediction errors satisfying $\|\mathbf{d} - \hatd\|_\infty \leq \gamma$ (where $\gamma$ is generic in $[0,1]$), and while this loss bounded, it is potentially quite large.\footnote{Their bound is $\gamma|FV|$ (Thm.~4.2, 4.3, \citep{assos2025alternates}). In practice, $|FV|$ is roughly 20-30, so if any $\tilde{d}_i$ is off by $10\%$, a quota can be violated by a factor of 2.} This large loss is unsurprising because their algorithm is not \textit{designed} to be robust to such errors --- it simply accepts $\hatd$ as the truth. 

This brings us to our first algorithmic goal: given $\hatd$, we want to compute a panel that does as well as possible on \textsc{I. Panel Size} and \textsc{II. Representation} after dropouts are drawn from the worst-case distribution with marginal dropout probabilities $\dvec : \|\dvec - \hatd\|_\infty \leq \gamma$. Here, $\gamma$ can then be interpreted as the \textit{robustness level} of such an algorithm, and can be chosen by the user. Importantly, while we constrain errors in \textit{marginal} dropout probabilities the same way as Assos et al., we aim to be robust against an even strong adversary: while Assos et al.~assume dropouts occur \textit{independently} according to $\hatd$, we allow \textit{arbitrary correlations} between panelists' dropout events.  

\textbf{Achieving \textsc{III. $(\alpha,\beta)$-Equality}.} 
\textit{ERM-Benchmark} does not consider the ideal of \textsc{Equality} at all. In fact, it is \textit{deterministic}: if there is a uniquely optimal set of panelists / alternates, their algorithm deterministically chooses it. As we confirm empirically in real data (\Cref{section:experiments}), their algorithm often gives pool members selection probability 1. In addition to robustness $\gamma$, we therefore want to allow the user to control the $(\alpha,\beta)$\textsc{-Equality} of the selection process. This means they can provide $\alpha,\beta \in [0,1]$, and we want to guarantee that all selection probabilities induced by our sortition algorithm are in these bounds.

\subsection{Approach and Contributions}
More formally, our problem is as follows: We are given inputs including the pool $N$ of size $n$, the desired panel size $k$, and the quotas; the estimated marginal dropout probabilities $\hatd$; the desired robustness level $\gamma \in [0,1]$; and equality requirements $\alpha, \beta \in [0,1]$. We want to output a panel $\tilde{K} \subseteq N$ such that after some set of agents $D$ drops out, the post-dropout panel $\tilde{K}\setminus D$ has minimal \textit{Loss}, i.e., only minimally violates the quotas and desired panel size. For the algorithm to be \textit{optimally} $\gamma$-robust, it must minimize the expected loss of $\tilde{K} \setminus D$ assuming that $D$ is sampled according to an adversarially chosen dropout distribution $\delta$. The only constraint on $\delta$ is that the marginal dropout probabilities $\mathbf{d}$ it implies are constrained so that $\dvec : \|\dvec - \hatd\|_\infty \leq \gamma$ --- i.e. its marginal dropout probabilities are in the $\gamma$-radius $L_\infty$ ball around $\hatd$. We let $\Delta(\hatd,\gamma)$ be the set of all joint distributions over dropout sets $D$ consistent with marginal probabilities $\dvec$ lying in the $\gamma$-ball around $\hatd$. The global minimax problem we would ideally like to solve is then 
\[\min_{\tilde{K} \subseteq N} \max_{\delta \in \Delta(\hatd,\gamma)} \mathbb{E}_{D \sim \delta}[Loss(\tilde{K} \setminus D)].\]

\textit{Key challenge 1.} The first challenge here is that the set of all possible dropout sets $D$ is exponential, meaning the inner expectation is a sum over \textit{exponentially many} possible events. Assos et al. contend with the same problem when optimizing $\min_{\tilde{K} \subseteq N} \mathbb{E}_{D \sim \hatd}[Loss(\tilde{K} \setminus D)]$ \emdash our problem without the inner max\emdash and they deal with it using ERM: they empirically estimate the distribution over $D$ via sampling (showing that this estimate is good enough via uniform convergence), and then using integer programming to select $\tilde{K}$. We cannot use ERM due to the added maximization step, because there is \textit{no fixed distribution to sample} --- we need to simultaneously optimize our panel (an integer problem) and the dropout distribution playing against it (a continuous problem).

\textbf{Contribution I: A sortition algorithm that is near-optimally $\gamma$-robust and $(\alpha,\beta)$\textsc{-Equal}.} We circumvent the intractability of this bilevel mixed-integer optimization problem by first solving its \textit{continuous relaxation}, and then dependently rounding the resulting fractional panel. Our algorithm thus proceeds in two steps: first, it solves the \textit{fractional version} of this problem optimally, and then it rounds this fractional panel into the panel $\tilde{K}$.

\textit{\textbf{Step 1: Compute a fractional panel.}} We compute the \textit{fractional panel} $\hatp \in [0,1]^n$ that solves the following minmax optimization problem.
\begin{center}
    $\min_{\pvec}\max_{\delta \in \Delta(\hatd,\gamma)} \mathbb{E}_{D \sim \delta} [\textit{Loss}(\pvec \setminus D)] \quad  \text{\textbf{s.t.}} \quad \pvec \in [\alpha,\beta]^n.$
\end{center} 
Here, $\pvec \setminus D$ represents the (fractional) panel post-dropout. Note that this is exactly the continuous relaxation (i.e., pool members are treated as divisible) of our overall set selection problem, up to one addition: we require that all agents are fractionally included to a degree within $[\alpha,\beta]$.

In \Cref{sec:algos}, we give an algorithm, \textsc{MinMax-Optimizer} (\Cref{alg:min_max_response}), for solving this minmax problem. This algorithm plays two subroutines against each other over $T$ rounds: in each round $t$, the \textsc{Minimizer} (\Cref{alg:pgd}) uses projected gradient descent to design a fractional panel $\pvec^t$, and then the \textsc{Maximizer} (\Cref{alg:ellipsoid_br}) best-responds with the distribution over dropout sets $\delta^t \in \Delta(\hatd,\gamma)$ that maximizes loss against $\pvec^t$. We prove that the \textit{average} fractional panel $\pvec^t$ produced over iterates $t \in [T]$ converges to the optimal fractional panel $\pvec^*$ at a rate $\sqrt{n/T}$ (\Cref{thm:no-regret}).

\textit{Key challenge 2.} The key technical difficulty is that the inner maximization problem over the dropout distribution $\delta$\emdash which must be solved by the \textsc{Maximizer} algorithm\emdash is defined over distributions with exponentially large support (all possible dropout sets $D$). This maximization problem can be expressed as a linear program with $O(n)$ constraints (one for each of the marginal probability upper and lower bounds), but $2^n$ variables (one for each dropout set $D$). 

We show that the inner maximization problem can nevertheless be solved efficiently by instead solving the \emph{dual} linear program (which has $O(n)$ variables but $2^n$ constraints) using the Ellipsoid algorithm with a polynomial-time separation oracle. Informally this is because our objective is the maximum of only polynomially many \emph{modular} functions in $D$, whose maximizer must be the maximizer of at least one of the constituent modular functions. The set-valued maximizer for a modular function can be computed as a simple thresholding, and our separation oracle amounts to finding the optimizer for each of the constituent modular components of the objective and then taking the maximum over them. Finally we show that that an optimal primal solution can be extracted from the optimal dual solution. Together this gives an efficient algorithm to solve the fractional minimax problem, despite the fact that the ``strategy space of the maximization player'' is exponentially large.

\textit{\textbf{Step 2: Dependently round.}} We then dependently round our optimal fractional panel $\hatp$ into the panel $\tilde{K}$. For this, we use a known polynomial-time dependent rounding procedure, \textit{Pipage Rounding} \citep{gandhi2006dependent}. In \Cref{section:formal-guarantees-new}, we apply \textit{Pipage}'s marginal preservation and negative association properties to prove bounds on our algorithm's achievement of ideals \textsc{I-III}.

\textbf{Contribution II: Real-world trade-offs and algorithm performance.} Our algorithms' optimality (at least in the fractional problem) allows us to explore trade-offs in this problem that were previously inaccessible. In \Cref{section:experiments}, we ask: \textit{Is it worse to be overly robust when $\tilde{d}$ is a good estimate, or non-robust when $\tilde{d}$ is a poor estimate?}, and \textit{What is the cost to robustness and loss of achieving stronger $(\alpha,\beta)$\textsc{-Equality}?} We characterize these trade-offs, finding that overestimating the power of the adversary (i.e., being too robust) tends to be \textit{costlier} than being not robust enough. We also find that tightening $[\alpha,\beta]$ has a roughly linear effect on the level of loss achievable. Finally, our algorithms' loss and that of \textit{ERM-Benchmark} against adversaries of different strengths (i.e., values of $\gamma$).

\subsection{Additional Related Work} \label{sec:relatedwork}
Our algorithmic approach is new to the sortition literature, but the peripheral tools and ideas we use track this literature closely. Ideals \textsc{I-III} have been pursued in various forms 
\citep{halpern2025federated,flanigan2020neutralizing,flanigan2021fair,flanigan2021transparent,ebadianboosting,EKMP+22,baharav2024fair,assos2025alternates}; the dependent rounding procedure has been applied to the sortition problem (though for rounding a different object) \citep{flanigan2021transparent}; prior work has also studied the continuous relaxation of the panel selection problem \citep{flanigan2020neutralizing,flanigan2024manipulation}; and column generation is a common approach in the sortition literature for designing distributions over intractable set spaces \cite{flanigan2021fair}.

The technique we use to solve the minimax optimization problem of playing a ``no-regret learning algorithm'' against a ``best-response algorithm'' originates with \cite{freund1999adaptive}. This core algorithmic technique has been widely used to solve constrained optimization problems with exponentially-many constraints (as in our case) --- see e.g. examples in differential privacy \citep{gaboardi2014dual}, algorithmic fairness \citep{kearns2018preventing}, and conditional calibration \citep{haghtalab2023unifying}. In all of these cases, the technique requires applying an appropriate no-regret learning algorithm (we use online gradient descent \citep{zinkevich2003online}) and then giving an efficient algorithm to solve a ``best response'' problem (in our case, our algorithm for computing the worst-possible dropout distribution given a fractional panel). The details of these solutions typically differ in which algorithmic techniques are used to solve the ``best response'' problem for the player that has the large strategy space: in our case this manifests itself in our solution to the dropout player's best response problem using Ellipsoid on the dual of their best response problem. 

The set of dropout distributions we solve the minimax optimization problem over are specified by linear moment constraints on marginals and expected size, which is related to strategies used for distributionally robust optimization (DRO). Classical and modern DRO results (moment sets, $\varphi$-divergences, and Wasserstein balls) typically study convex sample-average losses; in contrast, we have a  combinatorial objective defined via quotas; this is what necessitates our no-regret vs. best response style algorithm rather than using a closed-form dual reformulation. \citep{DelageYe2010,WiesemannKuhnSim2014,EsfahaniKuhn2018,RahimianMehrotra2019}.

\section{Model}
\label{section:model-new}
\textbf{Instance.} An instance of our problem $\inst=(N,k,\boldsymbol{\ell},\mathbf{u})$ consists of a \textit{pool} $N$ (a set of agents) of size $n$, a panel size $k\in\mathbb{N}_{>0}$, and \textit{quotas}. These quotas are imposed on \textit{feature-value} pairs, defined as follows: let the set of possible values of each feature $f \in F$ be $V_f$, and let
$\FV := \{(f,v): f\in F,\; v\in V_f\}$ be the set of all feature-value pairs. As an simple example, we could have feature set $F = \{\text{age, height}\}$ with $V_{\text{height}} = \{\text{short}, \text{tall}\}$; a feature-value pair would be (height,tall). Then, a pair of lower and upper quotas on $(f,v)$ ad defined respectively as $\ell_{f,v},u_{f,v} \in \mathbb{N}$, where $1 \leq \ell_{f,v} \leq u_{f,v}$. We summarize the quotas as $\boldsymbol{\ell} = (\ell_{f,v} | f,v \in FV)$ and $\mathbf{u} = (u_{f,v} | f,v \in FV)$.

Let $f : N \to V_f$ so that $f(i)$ is agent $i$'s value for feature $f$, and let $N_{f,v} := \{ i\in N : f(i)=v\}$ be the set of pool members with feature-value $(f,v)$. Then, a\textit{ panel} $K\subseteq N$ \textit{satisfies} the quotas if $|K\cap N_{f,v}|\in[\ell_{f,v},u_{f,v}]$ for all $(f,v)\in\FV$.

\textbf{Panel size as a quota.} We encode the panel size constraint via an extra feature $f^k$ with a single value $V_{f^k} = \{1\}$. We set $\ell_{f^k,1}=u_{f^k,1}=k$, and implicitly include $(f^k,1)$ in $\FV$ and its quotas in $\boldsymbol{\ell},\mathbf{u}$.

\textbf{Selection probabilities.} A \textit{selection algorithm} is any procedure for taking in an instance $(N,k, \boldsymbol{\ell},\mathbf{u})$ and outputting a panel $K$. Any selection algorithm induces \textit{selection probabilities} $\boldsymbol{\pi}=(\pi_i)_{i\in N}$, which we define such that $\pi_i=\Pr[i\in \tilde{K}]$ is the chance of being chosen for the \emph{initial} panel $\tilde{K}$.\footnote{We define selection probabilities as the chance of being chosen for the \textit{initial panel}, rather than being chosen \textit{and not dropping out}. This choice is most faithful to the known relationship between selection probabilities and manipulation \citep{flanigan2024manipulation}, and reflects the sortition principle of \textit{equality of opportunity} \citep{carson1999random}}

\subsection{Distributional Dropout Adversary}
\label{subsec:dist-adv}
Let $\hatd\in[0,1]^n$ be estimated per-agent dropout probabilities and fix a robustness radius $\gamma\in[0,1]$. Instead of assuming independent dropouts across agents, we allow the adversary to select an \emph{arbitrary dropout distribution} $\delta$ over dropout sets $D\subseteq N$, subject to linear marginal constraints:
\begin{align}
\label{eq:delta-feasible}
&\Delta(\hatd,\gamma) := \Big\{\,\delta\in\Delta(2^N):\; \sum_{i} \Pr[i\in D]=\sum_i \tilde d_i \ \land \notag \\
& \forall i, \Pr_{D\sim\delta}[i\in D]\in[\max(0,\tilde d_i-\gamma),\min(1,\tilde d_i+\gamma)] \Big\} 
\end{align}
The first equality constraint fixes the \emph{expected number} of dropouts in the pool to match $\hatd$; the second bounds the marginal deviation of $\dvec$ from $\hatd$ within a $\gamma$-ball. 

This formulation allows an adversary to optimize over the set of all (possibly correlated) distributions over dropouts that satisfy marginal dropout probabilities constrained to be within a $\gamma$-ball around the given marginal probability vector --- a more powerful adversary than one who optimizes over only independent Bernoulli dropouts. 
\subsection{Loss of a (Fractional) Panel}
For any fractional panel $\pvec$, dropout set $D\subseteq N$, and feature-value pair $(f,v) \in FV$, let $S_{f,v}$ be the number of fractional panelists with feature value pair $(f,v)$ post-dropout: 
\begin{align}
&S_{f,v}(\pvec;D) := \sum_{i\in N_{f,v}} p_i\, \mathbb{I}(i\notin D),
\intertext{Let the deviations of this post-dropout panel from the quotas on $(f,v)$, normalized by the upper quota be}
&\dev^{-}_{f,v}(\pvec;D) = \tfrac{1}{u_{f,v}}(\ell_{f,v} - S_{f,v}(\pvec;D)),\\
&\dev^{+}_{f,v}(\pvec;D) = \tfrac{1}{u_{f,v}}(S_{f,v}(\pvec;D) - u_{f,v}).
\end{align}
Now, let $h(\pvec;D)$ be the maximum normalized deviation across \textit{all} feature-values. Finally, fixing a dropout distribution $\delta$, let the overall \textit{loss} be the expectation of this maximum deviation over dropouts $D \sim \delta$:
\begin{align}
\label{eq:hD}
h(\pvec;D) &:= \max_{(f,v)\in\FV, \  \dagger \in \{-,+\}} \dev^{\dagger}_{f,v}(\pvec;D)\\
 \mathcal{L}(\pvec,\delta;\inst) &:= \mathbb{E}_{D\sim\delta}\big[\, h(\pvec;D)\,\big].\label{eq:new-loss}
\end{align}

\textbf{Action spaces.} The maximizer acts in $\Delta(\hatd,\gamma)$, as in from~\eqref{eq:delta-feasible}. We constrain the minimizer so they include each agent to a fractional degree within $[\alpha,\beta]$:
\begin{equation}
\label{eq:min-space}
 \mathfrak{P}(\alpha,\beta) := [\alpha,\beta]^n \subseteq [0,1]^n.
\end{equation}
One may want to additionally constrain the fractional panel size as $\|\pvec\|_1=k^+$, for some $k^+ \geq k$\emdash this does not affect the convexity. Our fractional optimization problem can now be expressed as the convex--concave saddle-point-problem:
\begin{equation}
\label{eq:opt-prob}
 \pvec^*:=\arg\min_{\pvec\in \mathfrak{P}(\alpha,\beta)}\; \max_{\delta\in \Delta(\hatd,\gamma)}\; \mathcal{L}(\pvec,\delta;\inst).
\end{equation}

\paragraph{Convexity--Concavity (proof in \Cref{app:convconc}).} For any fixed $\delta$, $\mathcal{L}(\cdot,\delta;\inst)$ is convex in $\pvec$ because $h(D;\pvec)$ is a pointwise maximum of finitely many affine forms in $\pvec$, and expectation preserves convexity. For any fixed $\pvec$, $\mathcal{L}(\pvec,\cdot;\inst)$ is affine (hence concave) in $\delta$ by linearity of expectation. Since $\mathfrak{P}(\alpha,\beta)$ and $\Delta(\hatd,\gamma)$ are convex and compact, Sion's minimax theorem then implies a saddle point exists and the max and min in \Cref{eq:opt-prob} can be exchanged.

\section{Our Optimization Algorithm }\label{sec:algos-new}\label{sec:algos}
We solve the minmax optimization problem in~\eqref{eq:opt-prob} by playing a no-regret minimizer against an exact best response of the maximizer. The minimizer runs projected subgradient descent; the maximizer solves a linear program by running Ellipsoid using an efficient separation oracle on its dual (Appendix~\ref{app:new-adv}).

\subsection{Algorithm}
Our overall algorithm is \Cref{alg:min_max_response}, and its subroutines are Algorithms \ref{alg:pgd} and \ref{alg:ellipsoid_br}.
Let $T\in\mathbb{N}$ and a stepsize $\eta>0$ be given. 

\begin{algorithm}[h]
\caption{\textsc{MinMax-Optimizer}}
\label{alg:minmax-cc}\label{alg:min_max_response}
\begin{algorithmic}[1]
\Require $\inst$, $\hatd$, $\gamma$, $\alpha,\beta$, $\eta$, $T$.
\State $\pvec^0 \leftarrow \{(\alpha+\beta)/2\}_{i=1}^n$; choose feasible $\delta^0\in\Delta(\hatd,\gamma)$
\For{$t=0,1,\dots,T-1$}
  \State $\pvec^{t+1} \leftarrow \textsc{Projected-Subgradient}(\pvec^{t},\delta^{t},\eta)$ 
  \State $\delta^{t+1} \leftarrow \textsc{Best-Response-Ellipsoid}(\pvec^{t+1},\hatd,\gamma)$ 
\EndFor
\State \Return $\hat{\pvec} := \frac{1}{T}\sum_{t=1}^T \pvec^t$
\end{algorithmic}
\end{algorithm}

\textbf{Minimizer:} \textsc{Projected-Subgradient.} Note that for fixed $\delta$, the map $\pvec\mapsto \mathcal{L}(\pvec,\delta;\inst)$ is convex and piecewise linear. A single gradient step proceeds as follows. For each dropout set in the support of $\delta$ we calculate the sub-gradient in accordance with the loss function, update $\pvec^t$ with the final gradient and clip the final probabilities to be in our action space $[\alpha, \beta]$. We defer the full details of this algorithm to the Appendix (\cref{alg:pgd}). 

\textbf{Maximizer:} \textsc{Best-Response-Ellipsoid}. \ For fixed $\pvec$, the best response problem for the adversary is to solve the following linear program:
\begin{align*}
\max_{\delta\ge 0}\; & \sum_{D\subseteq N} \delta(D)\, h(D;\pvec)\\
\text{s.t. } & \sum_{D\subseteq N} \delta(D)=1, \ \ \land \ \ \sum_{D\subseteq N} |D|\, \delta(D) = c  \land \ \ \sum_{D\ni i} \delta(D) \in [m_i, M_i] \; \forall i.
\end{align*}
with $m_i=\max(0,\tilde{d}_i-\gamma)$, $M_i=\min(1,\tilde{d}_i+\gamma)$,
The dual linear program (below) has the following variables, corresponding to the constraints in the primal: $\lambda\in\R$ (normalization), $\tau\in\R$ (expected size), and $a_i,b_i\ge 0$ (upper,lower marginals). 
\begin{align*}
\min_{\lambda,\tau,\mathbf{a},\mathbf{b}\ge 0}\; & \lambda + c\,\tau + \sum_i M_i a_i - \sum_i m_i b_i\\
\text{s.t. } & \lambda + \tau |D| + \sum_{i\in D} (a_i - b_i) \ge h(D;\pvec) \quad \forall D\subseteq N.
\end{align*}
We show that this best-response problem can be solved in polynomial time via the Ellipsoid method applied to the dual, using a polynomial-time separation oracle. This separation oracle evaluates $\max_{D\subseteq N}\{h(D;\pvec)-\sum_{i\in D} (a_i - b_i + \tau)\}$ in $O(|\FV|\,n)$ time by enumerating the $2|\FV|$ linear pieces and thresholding. Ellipsoid yields a dual optimum in polynomial time (Theorem~\ref{thm:ellipsoid-dual}); we then recover a primal optimum by solving the primal restricted to the polynomial set of subsets returned during Ellipsoid (Theorem~\ref{thm:primal-recovery}).

\begin{algorithm}[h]
\caption{\textsc{Best-Response-Ellipsoid}$(\pvec,\hatd,\gamma)$}
\label{alg:ellipsoid_br}
\begin{algorithmic}[1]
  \State Initialize collection of subsets $\mathcal{C} \leftarrow \emptyset$.
  \State Run Ellipsoid on the dual variables $(\lambda,\tau,\mathbf{a},\mathbf{b})$, using the separation oracle that returns a maximizing $D$ for $h(D;\pvec)-\sum_{i\in D} (a_i - b_i + \tau)$.
  \State Each time the oracle finds a violated constraint, add the returned $D$ to $\mathcal{C}$.
  \State After Ellipsoid terminates with a feasible dual, solve the primal LP restricted to variables $\{\delta(D): D\in\mathcal{C}\}$ subject to the original constraints.
  \State \Return the optimal restricted primal solution (extended by zeros off $\mathcal{C}$). This is optimal for the full primal (Theorem~\ref{thm:primal-recovery}).
\end{algorithmic}
\end{algorithm}

\subsection{Convergence and Efficiency}

We next prove two key properties of \Cref{alg:min_max_response}: it converges to optimality as $T$ grows (\Cref{thm:no-regret}), and it is computationally efficient (\Cref{Thm:poly-time}). 
\begin{theorem}[Optimality]\label{thm:no-regret} Fix algorithm inputs $\inst,\hatd,\alpha,\beta,\gamma,\eta,T$; let $\hatp$ be the output of \Cref{alg:min_max_response} on these inputs; and let $\pvec^*$ be the optimal solution for these inputs, as in \Cref{eq:opt-prob}. Then,
\[\max_{\delta\in\Delta(\hatd,\gamma)} \mathcal{L}(\hat{\pvec},\delta;\inst) \;\le\; \max_{\delta\in\Delta(\hatd,\gamma)} \mathcal{L}(\pvec^*,\delta;\inst) \;+
O\!\left(\sqrt{\tfrac{n}{T}}\right).
\]

\end{theorem}
\begin{proof}[Proof sketch]
Since $\mathcal{L}(\cdot,\delta)$ is convex and has bounded subgradients (each coordinate bounded by $1/u_{f,v}\le 1$), projected subgradient descent enjoys $O(\sqrt{n/T})$ regret. Combining standard no-regret-to-saddle-point arguments for convex--concave games (e.g., \cite{freund1999adaptive}) with the exact best responses yields the bound. See Appendix \ref{app:proof-no-regret} for the formal proof.
\end{proof}

\begin{theorem}[Efficiency]\label{Thm:poly-time} \textsc{MinMax-Optimizer} is polynomial time.
\end{theorem}
\begin{proof}[Proof Sketch]
Each \textsc{Projected-Subgradient} iteration computes $\pvec^{t+1}$ in $O(|\mathrm{supp}(\delta^t)|\,|\FV|\,n)$ time (where $\mathrm{supp}(\delta^t)$ must be polynomial-size because exactly one $D$ is added to the support for each Ellipsoid iteration). Then, \textsc{Best-Response-Ellipsoid} computes $\delta^{t+1}$ in time polynomial in $(n,|\FV|,L)$ via Ellipsoid on the dual using the $O(|\FV|\,n)$ separation oracle, followed by solving a polynomial-size restricted primal (Theorems~\ref{thm:ellipsoid-dual} and~\ref{thm:primal-recovery}). Hence, \textsc{MinMax-Optimizer} is polynomial time.
\end{proof}

\section{Guarantees on Ideals I--III} 
\label{section:formal-guarantees-new}\label{section:formal-guarantees}

\textbf{Our selection algorithm.} Our overall selection algorithm, which we call  \textsc{MinMax-Pipage}, takes as input an instance $\inst$ along with $\gamma,\alpha,\beta$. It first runs \textsc{MinMax-Optimizer} to produce an optimal fractional panel $\hatp$; then, we round $\hatp$ into an initial panel $\tilde{K}$ via \textit{Pipage Rounding}. At a high level, Pipage Rounding iteratively adjusts the fractional elements of $\hatp$ while preserving feasibility until all entries become integral, forming $\tilde{K}$ \citep{gandhi2006dependent}. We note its key properties when we apply them below.

We first bound the extent to which \textsc{MinMax-Pipage} satisfies \textsc{I.~Panel Size} and \textsc{II.~Representation} (\Cref{thm:quotas-new}). These ideals are evaluated on the post-dropout panel \(\tilde K\setminus D\) in expectation over \(D\sim\delta\), where \(\delta\in\Delta(\hatd,\gamma)\) is the adversarial distribution against $\hatp$. Formally, we give a high-probability bound (over the randomness of \textit{Pipage}) on the per-$(f,v)$ increase in quota deviation between $\pvec^*$, the optimal \textit{fractional} panel, and our integer panel $\tilde{K}\setminus D$. 

\begin{theorem}[Representation and Panel Size]\label{thm:quotas-new}
Fix inputs \(\inst,\hatd,\gamma,\alpha,\beta,\eta,T\), and let \(\tilde K\) be the (random) panel produced by \textsc{MinMax-Pipage}.
Fix any $\delta \in \Delta(\hatd,\gamma)$, any $(f,v) \in FV$, and any $z > 0$. With probability at least $1-2\exp\left(-\tfrac{2z^2}{\sigma^2_{f,v}}\right)$,
\begin{align*}
&\max_{\dagger \in \{-,+\}}\E_{D\sim\delta}\left[\dev^{\dagger}_{f,v}(\tilde{K};D)\right]
\leq \mathcal{L}(\pvec^*,\delta;\inst) + \epsilon + \tfrac{z}{u_{f,v}},
\end{align*}
\begin{center}
    where  $\sigma^2_{f,v} = \sum_{i \in N_{f,v}} \Pr_{D\sim\delta}[i \notin D]^2$, $\epsilon = O\left(\sqrt{\frac{n}{T}}\right)$.

\end{center}
\end{theorem}
\begin{proof}[Proof Sketch]
The proof proceeds in two steps. First, we bound the gap between $\pvec^*$ and $\hatp$ using \Cref{thm:no-regret}, thus giving the $\epsilon$ term. Then, to bound the gap between $\tilde{K}$ and $\hatp$, we apply two properties of \textit{Pipage Rounding}: it is \textit{marginal-preserving}, i.e., $\Pr[i\in\tilde K]=p_i$ and $\boldsymbol{\pi}=\hatp$; and the indicators $\mathbb{I}(i \in \tilde{K})$ are negatively associated \citep{gandhi2006dependent}. From these properties\emdash with some careful handling of the expectation over $\delta$\emdash we show that the number of panelists on $\tilde{K}$ with feature-value $(f,v)$ concentrates around its expectation, i.e., the fractional number of panelists in this group on $\hat{p}_i$. 
\end{proof}

Finally, $(\alpha,\beta)$-\textsc{Equality} is evaluated on the selection probabilities $\boldsymbol{\pi}^*$ implied by \textsc{MinMax-Pipage}. That $\boldsymbol{\pi}^*$ satisfies this property simply follows from the fact that \textit{Pipage Rounding} is marginal-preserving: for all $i$, $\pi^*_i = \hat{p}_i \in [\alpha,\beta]$, where the last step is due to the minimizer's action space (\Cref{eq:min-space}).
\begin{theorem}[Equality] \label{thm:equality}
 \textsc{MinMax-Pipage} satisfies  $(\alpha,\beta)$\textsc{-Equality}.
\end{theorem}

\section{Empirical Evaluation of Algorithms}
\label{section:experiments}

\textbf{Datasets.} We evaluate our algorithms on six real-world datasets, all corresponding to Australian citizens' assemblies. Each dataset contains an instance $\inst = (N,k,\boldsymbol{\ell},\mathbf)$. For each pool member $i \in N$, we have their feature-values $f,v \in FV$, plus two binary indicators: whether they were selected for the panel, and whether or not they dropped out conditional on being selected for the panel. See \Cref{app:data} for details.

\textbf{Algorithm inputs $\hatd,\gamma,\alpha,\beta,\eta,T$.} We estimate $\hatd$ for each instance by fitting the same \textit{independent action} model used in \cite{assos2025alternates}, which assumes each feature-value contributes independently to an agent's dropout probability. Our training data thus consists of agents' features (inputs) and observed dropout behavior (labels). For each instance, we train our model on three other instances (which three varies per instance). See \Cref{app:hatd} for details. We choose 
  $\alpha$ and $\beta$ jointly chosen to fall within a $\kappa \in (0,1]$ fraction of the distance between the average selection probability $k^+/n$ and the extremes 0 and 1, where $k^+ =  k + k \cdot \mathbb{E}_{D \sim \hatd}[|D|]/n$, a slight inflation of the target panel size to offset the average dropout rate. Formally, for $\kappa \in [0,1]$, 
  \[\alpha(\kappa) := (1-\kappa) \cdot k^+/n, \ \  \beta(\kappa) := k^+/n + (1-k^+/n) \kappa.\]
All experiments use $T = 1000$ and $\eta = \sqrt{n/T}$. We show empirical convergence in \Cref{app:convergence}.

\textbf{True dropout probabilities $\dvec$.} While in theory $\delta$ should be computed \textsc{Best-Response-Ellipsoid}, the Ellipsoid algorithm is impractical to implement. We thus restrict the space of dropout distributions to all \textit{product distributions}. Then, our maximizer only needs to consider vectors of marginal dropout probabilities for each pool member $\mathbf{d} \in [0,1]^n$. This restricted best-response problem has a simple, extremely efficient, optimal greedy solution we call \textsc{Greedy-Maximizer}. Details are in \Cref{app:br-greedy}.

In each instance, we sample dropouts independently according to $\dvec$, the output of \textsc{Greedy-Maximizer} with inputs $\inst,\hatd,\gamma',\pvec$. Here, $\gamma'$ reflects the \textit{true} errors in our predictions and $\pvec$ is the panel whose robustness is being evaluated. For \textsc{ERM Benchmark}, $\pvec$ is a 0/1 vector representing an integer panel, and for us $\pvec$ is $\hatp$.

\subsection{Fundamental Tradeoffs} 
We now examine two fundamental trade-offs. We consider only our algorithms, as there \textit{is no} prior work to compare to  here: our algorithm is the first to enable these tradeoffs' exploration. We analyze \textsc{MinMax-Optimizer}'s output $\hatp$ directly to more crisply observe the trade-offs without the noise of rounding. Here, $\gamma$ is the robustness level of our algorithm, and $\gamma'$ is the true power of the adversary.

\textbf{\textit{Robustness} versus \textit{Non-robustness}.} Here, we ask: what is the cost/benefit of being \textit{robust} when $\hatd$ is closer to the truth $\dvec$, versus being \textit{non-robust} when $\hatd$ is far from the truth $\dvec$? We explore this in Fig.~\ref{fig:heatmap} by examining $\mathcal{L}(\hatp,\dvec;\inst)$ at each combination of $\gamma,\gamma' \in \{0,0.15, 0.30, 0.45, 0.60, 0.75\}^2$. We show results are shown for one representative instance; the rest are in \Cref{app:more_heatmaps}. Across instances, we observe that the smallest value in each row lies on the diagonal; this indicates that ideally, $\gamma = \gamma'$ (i.e., we train with knowledge of the degree of errors). Given that $\gamma'$ may not actually be known, we can compare the losses on the left of the diagonal\emdash where we underestimate $\hatd$'s errors $(\gamma < \gamma')$\emdash to those on the right, where we overestimate it $(\gamma > \gamma')$. The leftward values tend to be much smaller, indicating that \textit{it is generally better to \textbf{underestimate} $\hatd$'s error than to overestimate it.}  

\begin{figure}[h]
    \centering
    \begin{minipage}{0.45\linewidth}
        \centering
        \includegraphics[width=\linewidth]{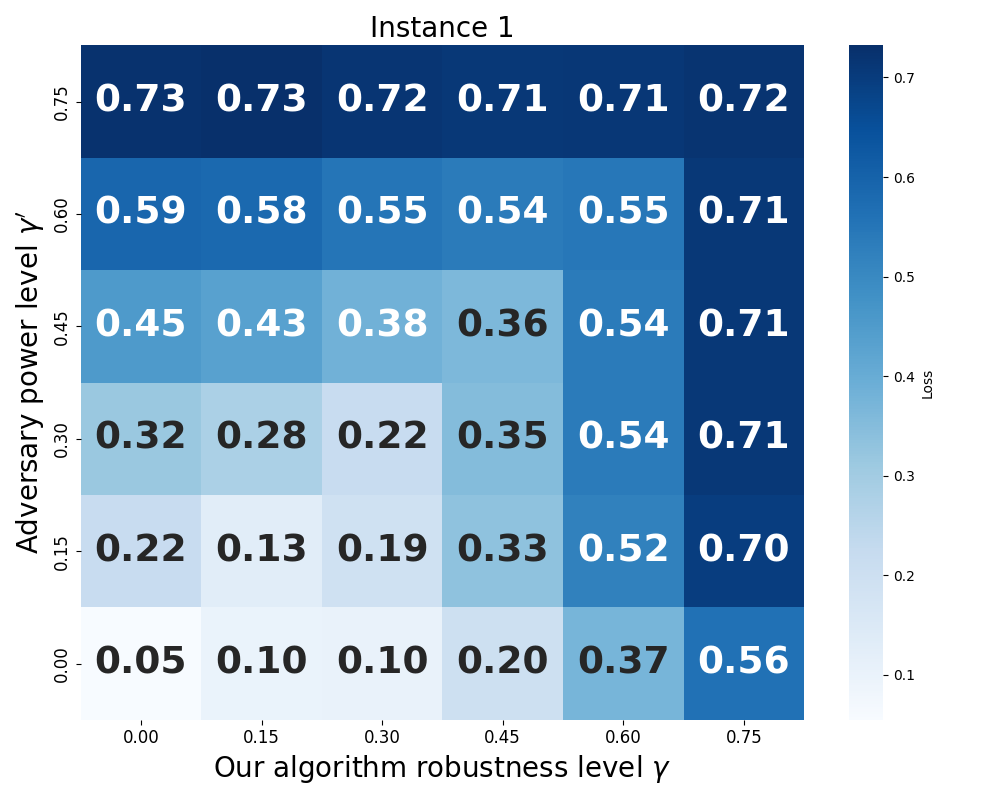}
        \caption{Here we show the loss of $\hatp$ at each combination of $\gamma$ and $\gamma'$ (the adversary's strength we train and test on, respectively) in Instance 1.}
        \label{fig:heatmap}
    \end{minipage}
    \hfill
    \begin{minipage}{0.45\linewidth}
        \centering
        \includegraphics[width=\linewidth]{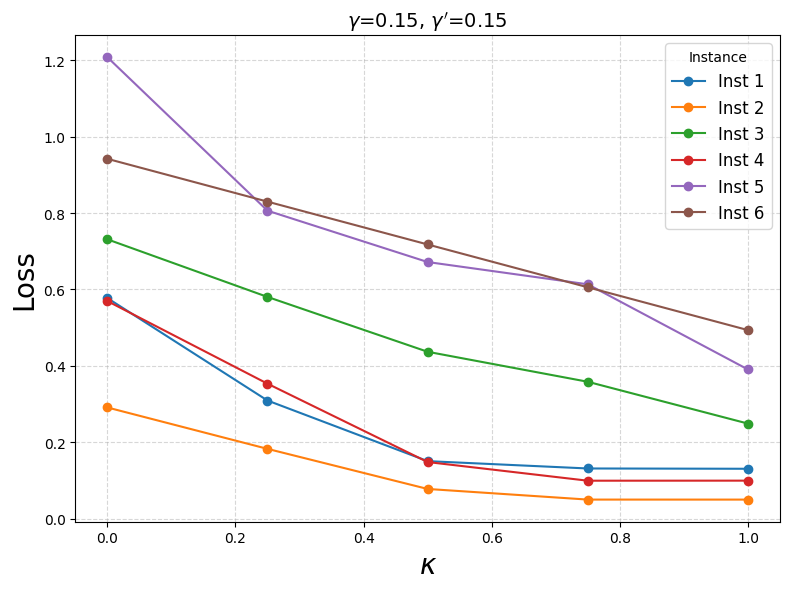}
        \caption{Here we show the loss of $\hatp$ over $\kappa \in \{0,0.25,0.5,0.75,1\}$, i.e., as the range $[\alpha,\beta]$ expands. Errors are chosen to be moderate: $\gamma = \gamma' = 0.15$.}
        \label{fig:kappa_vs_loss}
    \end{minipage}
\end{figure}

\textbf{Equality versus Loss.} Here, we ask: as we require stronger $\alpha,\beta$-equality, how much do we compromise on loss? We test this across instances by gradually increasing the selection probability gap $\kappa$, and thus widening the interval $[\alpha,\beta]$. On the vertical axis is loss $\mathcal{L}(\hatp,\dvec;\inst)$, where $\hatp$, $\dvec$ are derived with $\gamma = \gamma' = 0.15$. In Fig.~\ref{fig:kappa_vs_loss}, we see that across all instances, increasing $\kappa$ decreases the loss roughly linearly with some mild convexity. From $\kappa = 0$ (near equal probabilities) to $\kappa = 1$ (some probabilities equal to 0 or 1), the loss drops by 0.3-0.8, corresponding to the worst quota violation shrinking by 30\%-80\% of its upper quota's magnitude.

\subsection{Comparison of Algorithms on I - III} \label{sec:comparison}
We now compare our algorithm, \textsc{MinMax-Rounding}, to \textsc{ERM-Benchmark} \citep{assos2025alternates}, where \textsc{ERM-Benchmark} is configured to select an entire panel and modified to optimize a notion of loss comparable to ours. We describe our implementation of the benchmark in \Cref{app:erm-benchmark}.

\textsc{I. Panel Size:} In \Cref{tab:panel_size_comparison}, we compare (per instance) the ideal panel size $k$ with the expected post-dropout panel size $E_{D \sim \dvec}[|K \setminus D|]$ for several panels $K$: the fractional optimum $\hatp$, the panel produced by \textsc{MinMax-Pipage} $\tilde{K}$, and the panel produced by \textsc{ERM-Benchmark} $K^{\textsc{ERM}}$. Our algorithms are tested on dropout sets drawn from $\dvec$, the best-response to $\hatp$; \textsc{ERM-Benchmark} is tested on $\dvec$ defined as the best-response to $K^{\text{ERM}}$. We see that our algorithms do far better than \textsc{ERM-Benchmark} on panel size.

\begin{table}[h] 
    \centering
    \caption{Each column (except $k$) is implicitly the expectation over $D \sim \dvec$. Since $\tilde{K}$ is random over rounding, we sample $\tilde{K}$ 100 times and report the mean and standard error. $\kappa =1, \gamma = 0.15, \gamma' = 0.15$. Note that because panel size is encoded as a quota, deviations from optimal size trade off with deviations from other quotas.}
\label{tab:panel_size_comparison}
    \begin{tabular}{l| c c c c}
    \toprule
    \textbf{$\inst$} & \textbf{$k$} & $|\hatp \setminus D|$ & $|\tilde{K} \setminus D|$ & $|K^{\text{ERM}} \setminus D|$ \\
    \midrule
    1 & 51 & 46.2 & 46.5 $\pm$ 0.1 & 50.9 \\
    2 & 52 & 51.7 & 52.0 $\pm$ 0.1 & 54.5 \\
    3 & 41 & 40.3 & 40.5 $\pm$ 0.1 & 45.1 \\
    4 & 50 & 49.4 & 49.2 $\pm$ 0.1 & 54.6 \\
    5 & 47 & 49.5 & 49.9 $\pm$ 0.1 & 56.7 \\
    6 & 49 & 62.4 & 62.8 $\pm$ 0.1 & 67.8 \\
    \bottomrule
    \end{tabular}
\end{table}

\textsc{II. Representation:} In \Cref{fig:benchmark}, we compare the loss of our algorithm (both pre- and post-rounding) with that of \textsc{ERM-Benchmark} against increasingly adversarial dropouts $\gamma' \in \{0,0.3, 0.6, 0.75,1\}$ in Instance 1 (remaining instances are in \Cref{app:more_benchmark}). As expected, we see that our (fractional) algorithms eventually outperform \textsc{ERM-Benchmark} as $\gamma'$ grows, i.e., as robustness becomes more important. We also see that while our fractional panel $\hatp \setminus D$ usually outperforms or roughly matches the benchmark for low choices of $\gamma$, our rounded panel $\tilde{K} \setminus D$ has larger loss due to the cost of rounding fractional panels to integral ones via \textit{Pipage}. This cost gap is greater for instances with smaller and/or tighter quotas, where there is less slack within which to make fractional changes and such adjustments cause larger relative loss.

\begin{figure}[H]
    \centering
    \includegraphics[width=0.45\linewidth]{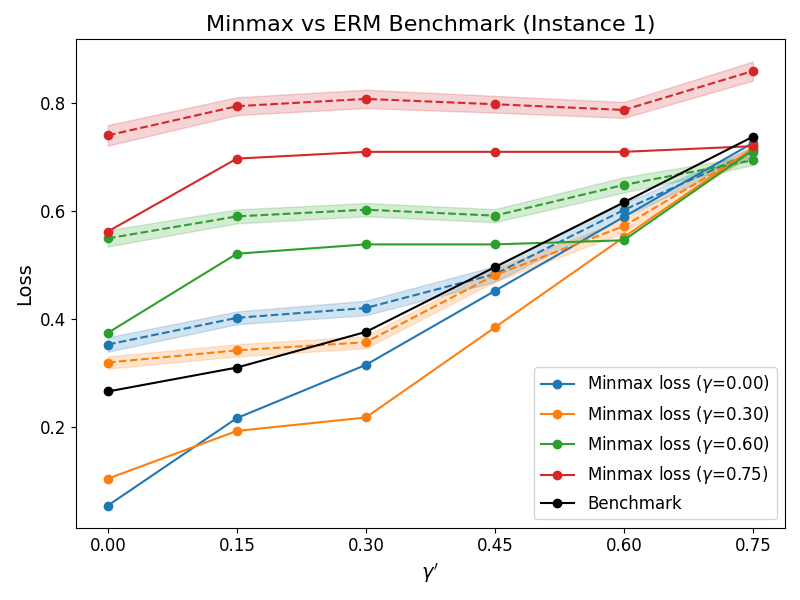}
    \caption{Here we compare $\mathcal{L}(\hatp,\dvec;\inst)$ (solid lines), $\mathcal{L}(\tilde{K},\dvec;\inst)$ (dotted lines) and $\mathcal{L}(K^{\text{ERM}},\dvec;\inst)$ (black line), over true error $\gamma'$ (x axis) and robustness level $\gamma$. For $\tilde{K}$, we show averages over 100 runs of \textsc{MinMax-Pipage}; shaded regions are standard errors.}
    \label{fig:benchmark}
\end{figure}

\textsc{III. Equality:} While our algorithms are $\alpha,\beta$-equal, we empirically show that \textsc{ERM-Benchmark} gives minimum and maximum selection probabilities of 0 and 1 respectively in all instances (see \Cref{app:selprobs}). 

\section{Future Work}
There remain several open problems on handling dropouts in the sortition process. The foremost is whether we can practicably (if not efficiently) solve the direct version of our optimization problem without the fractional relaxation, as formulated in the introduction. One could also formalize the differences between selecting \textit{alternates} versus \textit{entire robust panels}; while the latter is more convenient, alternates are in a sense fundamentally more ``powerful'', because they can be selectively deployed \textit{after} seeing the dropout set. Finally, studying the true ``predictability'' of dropout events in larger datasets may allow more realistic models of adversarial action spaces, leading to algorithms whose robustness is more appropriately targeted and possibly less costly when true prediction errors are low.

\bibliography{iclr-references}
\bibliographystyle{iclr2026_conference}

\newpage
\appendix
\section{Supplemental Materials from Sections 
\ref{section:model-new} and \ref{sec:algos-new}}
\subsection{\textsc{Projected-Subgradient}}
\label{app: Projected-Subgradient}
\setlength{\itemsep}{1.5pt}
\begin{algorithm}[h]
\caption{\textsc{Projected-Subgradient}$(\pvec^t,\delta^t,\eta)$}
\label{alg:pgd}
\begin{algorithmic}[1]
\State Initialize $g\leftarrow \mathbf{0}\in\mathbb{R}^n$
\For{each $D$ in the support of $\delta^t$}
    \State $(f^*,v^*,\dagger^*) \leftarrow \displaystyle \argmax_{(f,v)\in FV, \dagger \in \{-,+\}} \dev^{\dagger}_{f,v}(\pvec^t;D)$ 
    \If{$\dagger^*$ is $-$}
        \For{$i\in N_{f^*,v^*}$}
            \State $g_i \mathrel{+}= -\delta(D)\,\mathbb{I}(i\notin D)/u_{f^*,v^*}$
        \EndFor
    \ElsIf{$\dagger^*$ is $+$}
        \For{$i\in N_{f^*,v^*}$}
            \State $g_i \mathrel{+}= +\delta(D)\,\mathbb{I}(i\notin D)/u_{f^*,v^*}$
        \EndFor
    \EndIf
\EndFor
\State $\pvec^{t+1} \leftarrow \pvec^t - \eta\, g$
\State $\pvec^{t+1} \leftarrow \bigl(\min\{\max(p^{t+1}_i, \alpha), \beta\}\bigr)_{i=1}^n$
\State \Return $\pvec^{t+1}$
\end{algorithmic}
\end{algorithm}
\label{app:new-adv}

\subsection{Convexity--Concavity and Saddle Point} \label{app:convconc}
\begin{proposition}[Convexity in $\pvec$, concavity in $\delta$]\label{prop:convex-concave}
For fixed $\delta$, $\mathcal{L}(\cdot,\delta;\inst)$ is convex in $\pvec$; for fixed $\pvec$, $\mathcal{L}(\pvec,\cdot;\inst)$ is affine (hence concave) in $\delta$.
\end{proposition}
\begin{proof}
For any $D\subseteq N$, $S_{f,v}(\pvec;D)=\sum_{i\in N_{f,v}} p_i\,\mathbb{I}(i\notin D)$ is affine in $\pvec$. Therefore, $\phi_{f,v}(\pvec;D)=(1/u_{f,v})\,\max\{\ell_{f,v}-S_{f,v}(\pvec;D),\, S_{f,v}(\pvec;D)-u_{f,v}\}$ is a pointwise maximum of two affine forms, thus convex. The function $h(D;\pvec)=\max_{(f,v)\in FV} \phi_{f,v}(\pvec;D)$ is a pointwise maximum of finitely many convex functions, hence convex. Expectation preserves convexity, so $\pvec\mapsto \mathcal{L}(\pvec,\delta;\inst)=\E_{D\sim\delta}[h(D;\pvec)]$ is convex. For fixed $\pvec$, linearity of expectation implies that $\delta\mapsto \mathcal{L}(\pvec,\delta;\inst)$ is linear (affine) in the coordinates $\{\delta(D)\}_{D\subseteq N}$.
\end{proof}

\begin{theorem}[Sion's Minimax Theorem]
If $\mathfrak{P}(\alpha,\beta)$ and $\Delta(\hatd,\gamma)$ are convex and compact, and $\mathcal{L}$ is convex in $\pvec$ and concave in $\delta$, then 
\[
 \min_{\pvec\in\mathfrak{P}(\alpha,\beta)}\max_{\delta\in\Delta(\hatd,\gamma)} \mathcal{L}(\pvec,\delta;\inst) 
 \;=\; 
 \max_{\delta\in\Delta(\hatd,\gamma)}\min_{\pvec\in\mathfrak{P}(\alpha,\beta)} \mathcal{L}(\pvec,\delta;\inst).
\]
\end{theorem}

To verify that Sion's minimax theorem applies to our case, it remains to verify the convexity/compactness of the action spaces for the minimization and maximization players. $\mathfrak{P}(\alpha,\beta)$ is a product of closed intervals; $\Delta(\hatd,\gamma)$ is the intersection of the probability simplex with linear marginal constraints and (optionally) a linear expected-size equality; both are convex and compact. Combining with the previous proposition and applying Sion's minimax theorem, we see that we can interchange the min and the max in our formulation. As the $\max\min$ value of the game is equal to the $\min\max$ value of the game, henceforth we will refer to simply the \emph{value} of the game when referring to either quantity.

\subsection{Polynomial-Time Best Response via Ellipsoid}
For fixed $\pvec$, the best response problem for the  adversary is to solve the following linear program:
\begin{align}
\max_{\delta\ge 0}\; & \sum_{D\subseteq N} \delta(D)\, h(D;\pvec)\\
\text{s.t. } & \sum_{D\subseteq N} \delta(D)=1,\qquad \sum_{D\ni i} \delta(D) \in [\ell_i^{\mathrm{marg}}, u_i^{\mathrm{marg}}] \; (i=1,\dots,n),\\
& \sum_{D\subseteq N} |D|\, \delta(D) = c,
\end{align}
with $\ell_i^{\mathrm{marg}}=\max(0,\tilde{d}_i-\gamma)$, $u_i^{\mathrm{marg}}=\min(1,\tilde{d}_i+\gamma)$, and $c=\sum_i \tilde{d}_i$.
We now take the dual. The dual linear program has the following dual variables, corresponding to the constraints in the primal: $\lambda\in\R$ (normalization), $\tau\in\R$ (expected size), $a_i\ge 0$ (upper marginals), $b_i\ge 0$ (lower marginals). The dual is
\begin{align}
\min_{\lambda,\tau,\mathbf{a},\mathbf{b}\ge 0}\; & \lambda + c\,\tau + \sum_i u_i^{\mathrm{marg}} a_i - \sum_i \ell_i^{\mathrm{marg}} b_i\\
\text{s.t. } & \lambda + \tau |D| + \sum_{i\in D} (a_i - b_i) \ge h(D;\pvec) \quad \forall D\subseteq N.
\end{align}
We show that this best-response problem can be solved in polynomial time via the Ellipsoid method applied to the dual, using a polynomial-time separation oracle.

\paragraph{Dual separation oracle in $O(|\FV|\,n)$ time.} Define $w_i := a_i - b_i + \tau$. The family of dual constraints is equivalent to
\begin{equation}
\label{eq:dual-sep}
\max_{D\subseteq N} \Big\{ h(D;\pvec) - \sum_{i\in D} w_i \Big\} \;\le\; \lambda.
\end{equation}
To find a violated constraint (if one exists) it suffices to find $D$ maximizing the left hand side.
We now show the left-hand maximum can be evaluated in $O(|\FV|\,n)$ time.

\paragraph{Explicit decomposition and constants.} We use the fact that the function $h(D; \pvec)$ is the maximum of $2|\FV|$ functions, each of which are linear (modular) in $D$, in which each term in the maximum is defined by feature value pairs $(f,v)$ and a sign in $\{+,-\}$. The $D$ that maximizes this maximum must be the maximizer of a single one of the constituent modular functions. Optimizing a modular function is simple and reduces to computing a threshold rule. Our separation oracle iterates over each of the constituent modular functions, finds the maximizer for each one, and then amongst this set chooses the one that maximizes the overall objective. We now make this explicit. Recall:
\[
S_{f,v}(\pvec;D) = \sum_{i\in N_{f,v}} p_i\, \mathbb{I}(i\notin D) = \underbrace{\sum_{i\in N_{f,v}} p_i}_{=:M_{f,v}}\; -\; \sum_{i\in N_{f,v}\cap D} p_i.
\]
Then the two sides of the per-constraint deviation in \eqref{eq:hD} decompose as the maximum over the following two linear/modular functions in $D$:
\begin{align*}
\frac{1}{u_{f,v}}(\ell_{f,v} - S_{f,v}) &= \underbrace{\frac{\ell_{f,v} - M_{f,v}}{u_{f,v}}}_{=:\,C^{-}_{f,v}}\; +\; \sum_{i\in N_{f,v}\cap D} \frac{p_i}{u_{f,v}},\\
\frac{1}{u_{f,v}}(S_{f,v} - u_{f,v}) &= \underbrace{\frac{M_{f,v} - u_{f,v}}{u_{f,v}}}_{=:\,C^{+}_{f,v}}\; -\; \sum_{i\in N_{f,v}\cap D} \frac{p_i}{u_{f,v}}.
\end{align*}
Thus the final objective can be written as the maximimum over $2\cdot |\FV|$ such linear/modular functions in $D$: $h(D;\pvec)=\max_{(f,v)\in\FV}\{H^{-}_{f,v}(D),\,H^{+}_{f,v}(D)\}$ where
\[
H^{-}_{f,v}(D)= C^{-}_{f,v} + \sum_{i\in N_{f,v}\cap D} \tfrac{p_i}{u_{f,v}},\qquad
H^{+}_{f,v}(D)= C^{+}_{f,v} - \sum_{i\in N_{f,v}\cap D} \tfrac{p_i}{u_{f,v}}.
\]
\paragraph{Per-element gains and thresholding.} Each of these modular functions has a simple optimizer characterized by a threshold. For any fixed term $t\in\{(f,v,\text{-}),\,(f,v,\text{+})\}$, the objective in \eqref{eq:dual-sep} specialized to $t$ equals
\[
H_t(D) - \sum_{i\in D} w_i 
           = C_t + \sum_{i\in D} g_i^t,\quad \text{with}\;
\begin{cases}
g_i^{(f,v,\text{-})} = \mathbb{I}(i\in N_{f,v})\, \tfrac{p_i}{u_{f,v}} - w_i,\\
g_i^{(f,v,\text{+})} = -\mathbb{I}(i\in N_{f,v})\, \tfrac{p_i}{u_{f,v}} - w_i.
\end{cases}
\]

Hence the inner maximum is attained by the \emph{threshold set}
\[
D_t^* := \{ i\in N : g_i^t > 0 \},\qquad \text{value }= C_t + \sum_{i: g_i^t>0} g_i^t.
\]
Evaluating all $2|\FV|$ candidates and selecting the best one evaluates the left-hand side of \eqref{eq:dual-sep} in $O(|\FV|\,n)$ time, furnishing a separation oracle.
\paragraph{Bounding the dual region.} Standard LP bit-complexity bounds imply there exists an optimal dual solution with each coordinate bounded in magnitude by $2^{\mathrm{poly}(L)}$, where $L$ is the input bit-length (of $\ell,u,\tilde{d},\gamma,\pvec$). Intersecting the dual with this axis-parallel box preserves optimality and provides the initial ellipsoid for Ellipsoid.

\begin{theorem}[Dual optimization via Ellipsoid]
\label{thm:ellipsoid-dual}
Given $\pvec$ and the separation oracle above, the Ellipsoid method computes an $\varepsilon$-optimal feasible dual solution in time polynomial in $(n,|\FV|,L,\log(1/\varepsilon))$. Choosing $\varepsilon=2^{-\Theta(L)}$ yields an exact rational optimum.
\end{theorem}
\begin{proof}
The dual feasible set is a convex polyhedron contained in a box of polynomial bit-size. The separation oracle for constraints \eqref{eq:dual-sep} runs in $O(|\FV|\,n)$ time by the decomposition and thresholding argument above. The Ellipsoid method with polynomial-size initialization and polynomial-time separation yields an $\varepsilon$-optimal point in a number of iterations polynomial in the dimension and $\log(1/\varepsilon)$. Since coefficients have bit-length $\le L$, taking $\varepsilon=2^{-\Theta(L)}$ recovers an exact optimum.
\end{proof}

\paragraph{Primal recovery from violated constraints.} During the Ellipsoid run, whenever a constraint is violated the oracle returns a subset $D\subseteq N$. Let $\mathcal{C}$ be the collection of all such subsets gathered by Ellipsoid (its size is polynomial).

\begin{theorem}[Primal recovery]
\label{thm:primal-recovery}
Let $\mathcal{C}$ be as above. Consider the primal LP restricted to variables $\{\delta(D): D\in \mathcal{C}\}$ and the original constraints. An optimal solution to this restricted LP, extended by zeros outside $\mathcal{C}$, is optimal for the full primal.
\end{theorem}
\begin{proof}
Let $(\lambda,\tau,\mathbf{a},\mathbf{b})$ be the Ellipsoid-output dual, which is feasible for all constraints indexed by $D\in\mathcal{C}$ and satisfies \eqref{eq:dual-sep} for all $D\subseteq N$. Let $\delta^{\mathcal{C}}$ be an optimal solution of the restricted primal. By weak duality, its value is at most the dual value; by construction, all constraints corresponding to $\mathcal{C}$ are enforced. Suppose there existed a full primal solution $\tilde\delta$ with strictly larger objective. Then some positive mass must lie on a subset $D\notin\mathcal{C}$ with strictly positive reduced cost at optimality, contradicting that the final dual satisfies \eqref{eq:dual-sep}. Therefore $\delta^{\mathcal{C}}$ (extended by zeros) is optimal for the full primal.
\end{proof}

\begin{corollary}[Polynomial-time best response]
\label{cor:poly-br}
The adversary's best response can be computed in time polynomial in $(n,|\FV|,L)$ by running Ellipsoid on the dual with the oracle above and then solving the primal restricted to the polynomial set of columns collected during the run; a final oracle call certifies optimality.
\end{corollary}

\paragraph{Variant: inequality on expected size.} If the expected-size equality is replaced by $\sum_D |D|\,\delta(D) \le c$, then $\tau\ge 0$ in the dual and the separation becomes $\max_{D}\{h(D;\pvec) - \sum_{i\in D}(a_i-b_i) - \tau |D|\} \le \lambda$. The same oracle and analysis apply and the statements above remain valid.

\paragraph{Subgradient bounds.} For use in no-regret analysis, observe that each coordinate of a subgradient of $\mathcal{L}(\cdot,\delta;\inst)$ at $\pvec$ has magnitude at most $\max_{(f,v)} 1/u_{f,v}\le 1$ and arises from at most one active side at each $D$ in the support of $\delta$. Thus $\|g\|_2\le \sqrt{n}$ and standard $O(\sqrt{n/T})$ rates apply.

\subsection{Proof of Theorem~\ref{thm:no-regret}}
\label{app:proof-no-regret}

We restate the claim. Let $\{\pvec^t\}_{t=0}^{T-1}$ be generated by Algorithm~\ref{alg:minmax-cc} and define $\hat{\pvec}:=\frac{1}{T}\sum_{t=1}^{T}\pvec^t$. Then
\begin{equation}
\label{eq:thm-no-regret}
 \max_{\delta\in\Delta(\hatd,\gamma)}\, \mathcal{L}(\hat{\pvec},\delta;\inst)
 \;\le\; \min_{\pvec\in\mathfrak{P}(\alpha,\beta)}\max_{\delta\in\Delta(\hatd,\gamma)}\, \mathcal{L}(\pvec,\delta;\inst)
 \; + \; O\!\left(\sqrt{\tfrac{n}{T}}\right).
\end{equation}

\paragraph{Setup.} Recall that  $\delta^t\in\arg\max_{\delta\in\Delta(\hatd,\gamma)}\mathcal{L}(\pvec^t,\delta;\inst)$ for each round $t$ (i.e., a best response to $\pvec^t$). Define the sequence of convex functions on $\mathfrak{P}(\alpha,\beta)$
\[ f_t(\pvec) := \mathcal{L}(\pvec,\delta^t;\inst), \qquad t=0,\dots,T-1. \]
By Proposition~\ref{prop:convex-concave}, each $f_t$ is convex, and a subgradient $g_t\in\partial f_t(\pvec^t)$ is produced by Algorithm~\ref{alg:pgd}.

\paragraph{Projected subgradient regret.} Let $\Pi_{\mathfrak{P}}$ be Euclidean projection onto $\mathfrak{P}(\alpha,\beta)$. The projected subgradient update is $\pvec^{t+1}=\Pi_{\mathfrak{P}}(\pvec^{t}-\eta g_t)$. For any $\pvec\in\mathfrak{P}(\alpha,\beta)$, the standard regret bound for projected subgradient descent (e.g., \cite{zinkevich2003online}) yields
\begin{equation}
\label{eq:psgd-regret}
 \sum_{t=0}^{T-1} f_t(\pvec^t) - \sum_{t=0}^{T-1} f_t(\pvec)
 \;\le\; \frac{\|\pvec-\pvec^0\|_2^2}{2\eta} 
 \, + \, \frac{\eta}{2} \sum_{t=0}^{T-1} \|g_t\|_2^2.
\end{equation}
Since $\mathfrak{P}(\alpha,\beta)=[\alpha,\beta]^n$, its Euclidean diameter is at most $D:=\sqrt{n}\,(\beta-\alpha)\le \sqrt{n}$, whence $\|\pvec-\pvec^0\|_2\le D$ for all $\pvec$. By the subgradient bound above, $\|g_t\|_2\le G$ with $G\le \sqrt{n}$.

Choosing the optimal stepsize $\eta = D/(G\sqrt{T})$ in~\eqref{eq:psgd-regret} and dividing by $T$ yields
\begin{equation}
\label{eq:avg-regret}
 \frac{1}{T}\sum_{t=0}^{T-1} f_t(\pvec^t) - \frac{1}{T}\sum_{t=0}^{T-1} f_t(\pvec)
 \;\le\; \frac{DG}{\sqrt{T}} \,=\, O\!\left(\sqrt{\tfrac{n}{T}}\right).
\end{equation}

\paragraph{From regret to saddle-point suboptimality.} Let $\pvec^\star\in\arg\min_{\pvec\in\mathfrak{P}(\alpha,\beta)} \max_{\delta\in\Delta(\hatd,\gamma)} \mathcal{L}(\pvec,\delta;\inst)$ and denote the game value by $v^*$ (recall the $\max\min$ value of the game is equal to the $\min\max$ value of the game by Sion's minimax theorem). Using $\pvec=\pvec^\star$ in~\eqref{eq:avg-regret} gives
\begin{equation}
\label{eq:step1}
 \frac{1}{T}\sum_{t=0}^{T-1} \mathcal{L}(\pvec^t,\delta^t;\inst)
 \;\le\; \frac{1}{T}\sum_{t=0}^{T-1} \mathcal{L}(\pvec^\star,\delta^t;\inst) 
 \; + \; O\!\left(\sqrt{\tfrac{n}{T}}\right).
\end{equation}
By definition of best response, $\mathcal{L}(\pvec^t,\delta^t;\inst)=\max_{\delta}\mathcal{L}(\pvec^t,\delta;\inst)$. By convexity of $\pvec\mapsto \mathcal{L}(\pvec,\delta;\inst)$ and Jensen's inequality,
\begin{equation}
\label{eq:step2}
 \max_{\delta}\, \mathcal{L}(\hat{\pvec},\delta;\inst)
 \;=\; \max_{\delta}\, \mathcal{L}\!\left(\frac{1}{T}\sum_{t=0}^{T-1}\pvec^t,\delta;\inst\right)
 \;\le\; \frac{1}{T}\sum_{t=0}^{T-1} \max_{\delta}\, \mathcal{L}(\pvec^t,\delta;\inst)
 \;=\; \frac{1}{T}\sum_{t=0}^{T-1} \mathcal{L}(\pvec^t,\delta^t;\inst).
\end{equation}
Combining~\eqref{eq:step1} and~\eqref{eq:step2},
\begin{equation}
\label{eq:step3}
 \max_{\delta}\, \mathcal{L}(\hat{\pvec},\delta;\inst)
 \;\le\; \frac{1}{T}\sum_{t=0}^{T-1} \mathcal{L}(\pvec^\star,\delta^t;\inst) 
 \; + \; O\!\left(\sqrt{\tfrac{n}{T}}\right)
 \;\le\; \max_{\delta}\, \mathcal{L}(\pvec^\star,\delta;\inst) 
 \; + \; O\!\left(\sqrt{\tfrac{n}{T}}\right) 
 \;=\; v^* + O\!\left(\sqrt{\tfrac{n}{T}}\right).
\end{equation}
This is exactly the claim in Theorem~\ref{thm:no-regret}.

\section{Supplemental Materials from Section \ref{section:formal-guarantees-new}}
\subsection{Dependent Rounding: Negative Association (NA) and Concentration}
\label{app:rounding-na}

We gather concentration tools used in the representation guarantees.

\begin{lemma}[Negative association, marginal preservation of Pipage rounding]
\label{lem:na-pipage}
Let $\tilde K\subseteq N$ be the random panel produced by Pipage rounding applied to $\pvec\in[0,1]^n$. Then the indicator variables $X_i:=\mathbb{I}(i\in \tilde K)$ are \emph{negatively associated} (NA). In particular, for any two disjoint index sets $A,B\subseteq N$ and coordinatewise nondecreasing functions $f,g$, $\mathbb{E}[f((X_i)_{i\in A})\,g((X_i)_{i\in B})] \le \mathbb{E}[f((X_i)_{i\in A})]\,\mathbb{E}[g((X_i)_{i\in B})]$. Moreover, $\mathbb{E}[X_i]=p_i$ for all $i$.
\end{lemma}
\begin{proof}[Proof sketch]
See \citet{gandhi2006dependent} for result that Pipage rounding preserves marginals and induces NA.
\end{proof}

\begin{lemma}[Hoeffding bound for NA variables - see, e.g., \citep{dubhashi1996balls}]
\label{lem:weighted-hoeffding-na}
For generic set $S$, let $(X_i)_{i\in S}$ be NA, $X_i\in\{0,1\}$ with $\mathbb{E}[X_i]=p_i$. For fixed weights $a_i\in[0,1]$, define $Y=\sum_{i\in S} a_i X_i$ and $\mu=\mathbb{E}[Y]=\sum a_i p_i$. Then for all $z>0$,
\begin{align*}
\Pr[\,Y-\mu \ge z\,] &\le \exp\!\left( -\frac{2z^2}{\sum_{i\in S} a_i^2} \right), &
\Pr[\,\mu-Y \ge z\,] &\le \exp\!\left( -\frac{2z^2}{\sum_{i\in S} a_i^2} \right).
\end{align*}
\end{lemma}

\subsection{Proof of \Cref{thm:quotas-new}}
\label{app:pipage-affine-expected-proof}

\begin{theorem}[Affine expected-quota concentration for \textsc{MinMax-Pipage}]\label{app:quotas-round}
Fix inputs \(\inst,\hatd,\gamma,\alpha,\beta,\eta,T\), and let \(\tilde K\) be the (random) panel produced by \textsc{MinMax-Pipage}.
Fix any $\delta \in \Delta(\hatd,\gamma)$, any $(f,v) \in FV$, and any $z > 0$. with probability at least $1-2\exp\left(-\tfrac{2z^2}{\sigma^2_{f,v}}\right)$,
\begin{align*}
&\max_{\dagger \in \{-,+\}}\E_{D\sim\delta}\left[\dev^{\dagger}_{f,v}(\tilde{K};D)\right]
\leq \mathcal{L}(\pvec^*,\delta;\inst) + \epsilon + \tfrac{z}{u_{f,v}},
\end{align*}
where 
$\sigma^2_{f,v} = \sum_{i \in N_{f,v}} \Pr_{D\sim\delta}[i \notin D]^2$, $\epsilon = O\left(\sqrt{\frac{n}{T}}\right)$
\end{theorem}

\label{thm:pipage-affine-expected}

\begin{proof}
Let \(X_i := \mathbb{I}(i\in \tilde K)\in\{0,1\}\) be the indicator variables produced by \textsc{Pipage} rounding on \(\hatp\). By the properties of \textsc{Pipage} (Lemma~\ref{lem:na-pipage}):
(i) \(\mathbb{E}[X_i]=\hat p_i\) (marginals preserved), and
(ii) the family \((X_i)_{i\in N}\) is negatively associated (NA).

Fix \((f,v)\). For any \(D\subseteq N\),
\[
S_{f,v}(\tilde K;D)=\sum_{i\in N_{f,v}} X_i\,\mathbb{I}(i\notin D),
\qquad
S_{f,v}(\hatp;D)=\sum_{i\in N_{f,v}} \hat p_i\,\mathbb{I}(i\notin D).
\]
Taking expectation over \(D\sim\delta\) (independent of \(\tilde K\)) and using \(a_i=\Pr[i\notin D]\),
\[
\mathbb{E}_{D}[S_{f,v}(\tilde K;D)\mid \tilde K] = \sum_{i\in N_{f,v}} a_i X_i
=: Y_{f,v}(X),
\qquad
\mathbb{E}_{D}[S_{f,v}(\hatp;D)] = \sum_{i\in N_{f,v}} a_i \hat p_i =: \mu_{f,v}.
\]
Since \((X_i)\) are NA 0–1 and the weights \(a_i\in[0,1]\), the NA weighted Hoeffding inequality (Lemma~\ref{lem:weighted-hoeffding-na}) applied to the linear form \(Y_{f,v}(X)=\sum_{i\in N_{f,v}} a_i X_i\) with variance proxy \(\sigma_{f,v}^2=\sum_{i\in N_{f,v}} a_i^2\) yields, for any \(z>0\),
\[
\Pr\big[\,|Y_{f,v}(X)-\mu_{f,v}| \ge z\,\big] \;\le\; 2\exp\!\Big(-\tfrac{2 z^2}{\sigma_{f,v}^2}\Big).
\]
Therefore, by the definitions of \(\dev^-_{f,v}\) and \(\dev^+_{f,v}\),
\[
\mathbb{E}_{D}\left[\dev^-_{f,v}(\tilde K)\right] - \mathbb{E}_{D}\left[\dev^-_{f,v}(\hatp)\right]
= \mathbb{E}_{D}\Big[\tfrac{\ell_{f,v}-S_{f,v}(\tilde K;D)}{u_{f,v}} - \tfrac{\ell_{f,v}-S_{f,v}(\hatp;D)}{u_{f,v}} \,\Big|\, \tilde K\Big]
= \tfrac{\mu_{f,v}-Y_{f,v}(X)}{u_{f,v}},
\]
\[
\mathbb{E}_{D}\left[\dev^+_{f,v}(\tilde K)\right] - \mathbb{E}_{D}\left[\dev^+_{f,v}(\hatp)\right]
= \mathbb{E}_{D}\Big[\tfrac{S_{f,v}(\tilde K;D)-u_{f,v}}{u_{f,v}} - \tfrac{S_{f,v}(\hatp;D)-u_{f,v}}{u_{f,v}} \,\Big|\, \tilde K\Big]
= \tfrac{Y_{f,v}(X)-\mu_{f,v}}{u_{f,v}}.
\]
Hence both absolute deviations are equal to \(|Y_{f,v}(X)-\mu_{f,v}|/u_{f,v}\).

 By dividing both sides of $|Y_{f,v}(X)-\mu_{f,v}| \ge z$ by $u_{f,v}$ we get,
\[
|\mathbb{E}_{D}\left[\dev^-_{f,v}(\tilde K)\right] - \mathbb{E}_{D}\left[\dev^-_{f,v}(\hatp)\right]| \le \tfrac{z}{u_{f,v}},
\qquad
|\mathbb{E}_{D}\left[\dev^+_{f,v}(\tilde K)\right] - \mathbb{E}_{D}\left[\dev^+_{f,v}(\hatp)\right]| \le \tfrac{z}{u_{f,v}}.
\]

Moreover, since \(\mathbb{E}_{D}\left[\dev^-_{f,v}(\hatp)\right]\le \mathcal{L}(\hatp, \delta ; \inst)\) and \(\mathbb{E}_{D}\left[\dev^+_{f,v}(\hatp)\right]\le \mathcal{L}(\hatp, \delta ; \inst)\,]\), and Theorem~\ref{thm:no-regret} gives $$\mathcal{L}(\hatp, \delta ; \inst) \le \max_{\delta\in\Delta(\hatd,\gamma)} \mathcal{L}(\pvec^*,\delta;\inst) + O(\sqrt{n/T})$$
This is equivalently a bound on arbitrary $\delta\in\Delta(\hatd,\gamma)$, as the RHS can only grow with the removal of the maximum, 
$$\implies \mathcal{L}(\hatp, \delta ; \inst) \le  \mathcal{L}(\pvec^*,\delta;\inst) + O(\sqrt{n/T})$$

Define $v^* := \mathcal{L}(\pvec^*,\delta;\inst)$. We can write the following equivalent probability statements and subsequently apply \cref{lem:weighted-hoeffding-na}:

\begin{align*}
    \Pr\!\Big[
\;|\mathbb{E}_{D}\left[\dev^-_{f,v}(\tilde K)\right] - \mathbb{E}_{D}\left[\dev^-_{f,v}(\hatp)\right]| \le \tfrac{z}{u_{f,v}}
\;\;
\Big] = \Pr\!\Big[\mathbb{E}_{D}\left[\dev^+_{f,v}(\tilde K)\right]\;\le\; \mathcal{L}(\pvec^*,\delta;\inst) + O\!\big(\sqrt{n/T}\big) + \tfrac{z}{u_{f,v}}
\Big]\\
\geq 1 -  2\exp\!\Big(-\tfrac{2 z^2}{\sigma_{f,v}^2}\Big)
\end{align*}

\begin{align*}
    \Pr\!\Big[\;\;
|\mathbb{E}_{D}\left[\dev^+_{f,v}(\tilde K)\right] - \mathbb{E}_{D}\left[\dev^+_{f,v}(\hatp)\right]| \le \tfrac{z}{u_{f,v}}
\Big] = \Pr\!\Big[\mathbb{E}_{D}\left[\dev^+_{f,v}(\tilde K)\right]\;\le\; \mathcal{L}(\pvec^*,\delta;\inst) + O\!\big(\sqrt{n/T}\big) + \tfrac{z}{u_{f,v}}
\Big]\\ \geq 1 -  2\exp\!\Big(-\tfrac{2 z^2}{\sigma_{f,v}^2}\Big)
\end{align*}

Combining these statements and taking a maximum over the direction $\dagger \in \{+,-\}$, we arrive at the theorem statement for any $\delta \in \Delta(\hatd,\gamma)$, any $(f,v) \in FV$, and any $z > 0$:

\begin{align*}
&\Pr \left[\max_{\dagger \in \{-,+\}}\E_{D\sim\delta}\left[\dev^{\dagger}_{f,v}(\tilde{K};D)\right]
\leq \mathcal{L}(\pvec^*,\delta;\inst) + O\left(\sqrt{\frac{n}{T}}\right) + \tfrac{z}{u_{f,v}} \right] \geq 1 -  2\exp\!\Big(-\tfrac{2 z^2}{\sigma_{f,v}^2}\Big),
\end{align*}

\end{proof}

\textbf{Remark:} The rounding procedure is deterministically sum-preserving, so if one has constrained $\|\hatp\|_1 = k^+$, $\tilde K$ is size $k^+$ deterministically, and $\mathbb{E}_{D \sim \delta}[|\tilde{K} \setminus D|] = k^+ - \mathbb{E}_{D \sim \delta}[|D \cap \tilde{K}|]$.

\newpage
\section{Supplemental Materials from \Cref{section:experiments}}
\subsection{Data Information} \label{app:data}
Our data consists of six instances of previously occurred panels. Each instance has columns for the individual features, inclusion in the panel (\textsc{selected}) and whether they dropped out(\textsc{is\_dropout}) if they \textit{were} selected. Instances 1,3,5,6 had predetermined, practitioner-constructed quotas. Both Instance 2 and 4 did not have practitioner constructed quotas, and we manually constructed a set to reflect the makeup of the \textit{original panel}. F each feature-value pair, we counted the number of people selected for that panel with that feature-value pair ($n_{f,v}$), and set the associated quota to be $(\ell_{f,v}, u_{f,v}) = n_{f,v} \pm 2$. In the table below, we report statistics on the size of the pool $n$, desired size of the panel $k$ and number of distinct feature value pairs $|FV|$. 

\begin{table}[h]
\centering
\begin{tabular}{|c|c|c|c|}
\hline
$\inst$ & $n$ & $k$ & $|FV|$ \\ \hline
1       & 250 & 51  & 17     \\ \hline
2       & 125 & 52  & 12     \\ \hline
3       & 133 & 41  & 22     \\ \hline
4       & 290 & 50  & 12     \\ \hline
5       & 518 & 47  & 11     \\ \hline
6       & 427 & 49  & 28     \\ \hline
\end{tabular}
\end{table}

\subsection{Greedy Best Response over Product Distributions}
\label{app:br-greedy}
We formally define the modified action space from \cref{section:experiments} below. The key alteration being is the change from a distribution over dropout sets $D \subseteq N$ to a distribution over marginal inclusion probabilities $\dvec \in [0,1]^n$. This represents a strict subset of the original action space in our theoretical formulation, and consequently, the best response on this more constrained set is only an \textit{approximation} to the best response problem over \textbf{all} distributions over subsets. The gain is that the algorithm is substantially more efficient to run compared to our exact Ellipsoid-based best response algorithm, and still leads to strong experimental results.

\textbf{The (new) action space of the maximizer.} As in \citet{assos2025alternates}, we model the dropouts from any panel $\tilde{K} \subseteq N$ as the result of independent Bernoulli coin-flips per agent $i \in \tilde{K}$, with agent-specific rate $d_i$. $d_i$ is then $i$'s \textit{dropout probability}. For mathematical convenience, we assume that every \textit{pool member} has an associated dropout probability, summarized as $\mathbf{d} = (d_i | i \in N)$. We define the \textit{dropout set} $D \subset \tilde{K}$ as the random set of dropouts; we abuse notation slightly and express this random process as $D \sim \mathbf{d}$. 

Given inputs $\hatd$ and $\gamma$, the maximizer's action space is
\begin{equation}
    \mathfrak{D}(\hatd,\gamma):= [0,1]^n \ \cap \ B^\infty_\gamma(\hatd) \ \cap \ \{\dvec : \|\dvec\|_1 = \|\hatd\|_1\},
\end{equation}
where $B^\infty_\gamma(\hatd)$ is the $\gamma$-radius $L_\infty$-norm ball around $\hatd$. The first set requires $\dvec$'s entries to be valid probabilities; the second is the requirement specified by our notion of robustness; and the third is a weak requirement to avoid the adversary gaining power by dropping out a different expected number of \textit{pool members} (not panelists) than would $\hatd$.

We present an updated loss function of $\pvec$ on dropout distribution $\dvec$ and the associated optimization problem: 

\begin{equation}
    \mathcal{L}(\pvec,\dvec;\inst) := \max_{f,v \in FV} \mathbb{E}_{D \sim \dvec}[\mathcal{L}_{f,v}(\pvec \setminus D;\inst)] 
\end{equation}

\begin{equation} \label{eq:opt-prob-emp}
   \pvec^*(\inst,\hatd,\gamma,\alpha,\beta):= \text{arg}\min_{\pvec \in \mathfrak{P}(\alpha,\beta)} \ \max_{\dvec \in \mathfrak{D}(\hatd,\gamma)}  \ \mathcal{L}(\pvec, \dvec;\inst) 
\end{equation}

This new setup allows for a simple, greedy maximizer best response algorithm. The algorithm works as follows: for each feature value pair $f,v$, the algorithm calculates the result of \textit{dropping out} the maximum possible number of pool members with $f,v$, or conversely \textit{retaining} the maximal number of pool members with $f,v$, each time projecting results with \textsc{Project-Dropout-Probs}. The best response to the panel probabilities $\pvec$ is the argument maximum of the loss-maximizing $\dvec$ for each $f,v$ pair and upper and lower quota violation. 

\begin{algorithm}[H]
\caption{\textsc{Greedy-Best-Response}}
\label{alg: greedy_br}
\begin{algorithmic}[1]
\Require $\inst$, $\hatd$, $\gamma$, $\mathbf{p}$

\State $\dvec^{min} \leftarrow \max(0, \hatd - \gamma)$, $\dvec^{max} \leftarrow \min(1, \hatd + \gamma)$

\For{each $f,v \in FV$}\vspace{0.25em}
    \State Initialize $N_{f,v} \leftarrow \{i : f(i) = v\}\ \text{sorted in ascending order of } p(i)$, 
    $\overline{N}_{f,v} \leftarrow N \setminus N_{f,v}$, 
    \vspace{0.25em}
    \State $\mathbf{d}^{\text{keep}} \leftarrow \mathbf{0}$,  
    $\mathbf{d}^{\text{drop}} \leftarrow \mathbf{0}$ \vspace{0.25em}
    \State $d_i^{\text{keep}} \leftarrow d^\text{min}_i$ for all $i \in N_{f,v}$    \Comment{\textbf{Strategy 1}: keep as many people with $f(i) = v$ as possible}
\vspace{0.25em}
    \State $d_i^{\text{keep}} \leftarrow d^\text{max}_i$ for all $i \in \overline{N}_{f,v}$ \vspace{0.25em}
    \State $\bar{\mathbf{d}}^{\text{keep}} \leftarrow \textsc{Project-Dropout-Probs}(\inst, \hatd, \mathbf{d}^{\text{keep}}, \dvec^{min}, \dvec^{max}, \overline{N}_{f,v} + N_{f,v})$
    \State $d_i^{\text{drop}} \leftarrow d^\text{max}_i$ for all $i \in N_{f,v}$    \Comment{\textbf{Strategy 2}: drop as many people with $f(i) = v$ as possible}
\vspace{0.25em}
    \State $d_i^{\text{drop}} \leftarrow d^\text{min}_i$ for all $i \in N_{f,v}$\vspace{0.25em}
    \State $\bar{\mathbf{d}}^{\text{drop}} \leftarrow \textsc{Project-Dropout-Probs}(\inst, \hatd, \mathbf{d}^{\text{drop}}, \dvec^{min}, \dvec^{max}, \overline{N}_{f,v} + N_{f,v})$ \vspace{0.25em}
    
    \State $\dvec^*_{f,v} \leftarrow \text{argmax}_{\dvec \in \{\bar{\mathbf{d}}^{\text{keep}},\bar{\mathbf{d}}^{\text{drop}}\}} \mathcal{L}(\pvec,\dvec)$

\EndFor
\State \Return argmax$_{f,v \in FV}\mathbf{d}^*_{f,v}$
\end{algorithmic}
\end{algorithm}

\begin{algorithm}[H]
\caption{\textsc{Project-Dropout-Probs}}
 \label{alg:greedyproject}
\begin{algorithmic}[1]
\Require $\inst, \hatd, \mathbf{d}^{\text{keep}}, \dvec^{min}, \dvec^{max}, \overline{N}_{f,v} + N_{f,v}$
\State $\text{diff} \leftarrow \sum_{i=1}^n \hatd_i - \sum_{i=1}^n d_i$
\If{$\text{diff} = 0$} \Return $\mathbf{d}$ \EndIf
\For{each $i \in \overline{N}_{f,v}, \text{ then } N_{f,v}$}
    \If{$\text{diff} < 0$}
        \State $\text{adj} \leftarrow \min(\text{diff}, d_i - \dvec^{min}_i)$
        \State $d_i \leftarrow d_i - \text{adj}$
    \Else
        \State $\text{adj} \leftarrow \min(-\text{diff}, \dvec^{max}_i - d_i)$
        \State $d_i \leftarrow d_i + \text{adj}$
    \EndIf
    \If{$\sum_{i=1}^n \hatd_i - \sum_{i=1}^n d_i = 0$} \textbf{break} \EndIf
\EndFor
\State \Return $\mathbf{d}$

\end{algorithmic}
\end{algorithm}

\begin{theorem}[Runtime of \textsc{Greedy-Best-Response}]
Algorithm~\ref{alg: greedy_br} (\textsc{Greedy-Best-Response}) runs in time  $O(|\FV| \cdot n)$.
\end{theorem}

\begin{proof}
For each feature–value pair $(f,v) \in \FV$, the algorithm:
\begin{enumerate}
    \item Constructs candidate dropout vectors 
    $\mathbf{d}^{\text{keep}}$ and $\mathbf{d}^{\text{drop}}$ in $O(n)$ time;
    \item Projects each via \textsc{Project-Dropout-Probs}, which iterates once through all agents and hence also runs in $O(n)$ time;
    \item Evaluates the loss function $\mathcal{L}(\mathbf{p}, \mathbf{d})$ twice, costing $O(n)$ total.
\end{enumerate}
Thus, the total work per feature–value pair is $O(n)$, and iterating over all $|\FV|$ such pairs gives 
$O(|\FV| \cdot n)$ overall.  
\end{proof}

\begin{theorem}[Correctness of \textsc{Greedy-Best-Response}]
Fix $\inst, \gamma, \hatd$ and selection probabilities $\pvec$. Algorithm~\ref{alg: greedy_br} returns $\dvec$ such that $$\dvec = \argmax_{\dvec \in \mathfrak{D}(\hatd,\gamma)} \mathcal{L}(\pvec, \dvec \; \inst)$$.
\end{theorem}

\begin{proof}
We prove that the \textsc{Greedy-Best-Response} algorithm (with the projection subroutine run using the index order ``complement first, then group members in ascending order of $p_i$'') returns a global maximizer of
\[
\mathcal L(\pvec,\dvec)
\;=\;
\max_{(f,v)\in\FV}\; \max\!\left\{\,\frac{\ell_{f,v}-S_{f,v}(\pvec,\dvec)}{u_{f,v}},\; \frac{S_{f,v}(\pvec,\dvec)-u_{f,v}}{u_{f,v}}\,\right\},
\]
where
\[
S_{f,v}(\pvec,\dvec) \;=\; \sum_{i\in N_{f,v}} p_i(1-d_i).
\]

Let $\dvec^\star$ be the vector returned by the algorithm. Suppose, for contradiction, that there exists a feasible $\dvec'\in\mathfrak{D}(\hat\dvec,\gamma)$ with
\[
\mathcal L(\pvec,\dvec') \;>\; \mathcal L(\pvec,\dvec^\star).
\]
Then there must exist some group $g=(f,v)$ such that the group term
\[
G_g(\dvec):=\max\Big\{\frac{\ell_g-S_g(\dvec)}{u_g},\frac{S_g(\dvec)-u_g}{u_g}\Big\}
\]
satisfies \(G_g(\dvec')>G_g(\dvec^\star)\). We analyze the two exhaustive cases: overshooting the upper quota and undershooting the lower quota.

\paragraph{Preliminary: we give maximal slack to \(N_g\).}  
Because \(G_g(\dvec)\) depends only on coordinates \(\{d_i:i\in N_g\}\), any feasible changes to coordinates outside \(N_g\) do not affect the group term. Therefore, without loss of generality, we may first modify \(\dvec'\) on \(\overline N_g\) to allocate as much allowable probability as possible there (respecting the box constraints) so that \(N_g\) is left with maximal slack and still consistent with feasibility. This is exactly the adjustment the algorithm performs when forming its candidates: it sets complement coordinates to their upper (or lower) bounds as needed before touching group members. Doing this cannot decrease \(G_g(\dvec')\), and it simplifies the comparison to $\dvec^\star$ because any remaining difference between $\dvec'$ and $\dvec^\star$ will lie inside \(N_g\).
\paragraph{Case 1: $\dvec'$ violates the upper quota more than $\dvec^\star$.}
That is,
\[
G_g(\dvec') = \frac{S_g(\dvec')-u_g}{u_g} > G_g(\dvec^\star),
\]
so
\[
S_g(\dvec') \;>\; S_g(\dvec^\star),
\quad\text{equivalently}\quad
\sum_{i\in N_g} p_i d_i' \;<\; \sum_{i\in N_g} p_i d_i^\star.
\]
Let $i^\star$ be an index in $N_g$ with the smallest $p_i$ for which $d_{i^\star}' < d_{i^\star}^\star$ (such an index exists because the weighted sum under $\dvec'$ is strictly smaller). By the preliminary reduction we may assume the complement coordinates in $\dvec'$ are already set to give maximal slack to \(N_g\), so any further decrease of $d_{i^\star}$ would require compensating increases somewhere (either within \(N_g\) on other indices or in the complement if slack remained). But the algorithm's `keep' construction already sets complement coordinates to their allowed maxima and sets group coordinates to their allowed minima before projection. Under the complement-first, ascending-$p$ projection the algorithm only increases group coordinates in order of increasing $p_i$, so it produces the allocation that minimizes $\sum_{i\in N_g} p_i d_i$ (equivalently maximizes $S_g$) subject to feasibility. Hence no feasible $\dvec'$ can have strictly smaller $\sum_{i\in N_g} p_i d_i$ than $\dvec^\star$, contradicting the assumption.

\paragraph{Case 2: $\dvec'$ violates the lower quota more than $\dvec^\star$.}
That is,
\[
G_g(\dvec') = \frac{\ell_g-S_g(\dvec')}{u_g} > G_g(\dvec^\star),
\]
so
\[
S_g(\dvec') \;<\; S_g(\dvec^\star),
\quad\text{equivalently}\quad
\sum_{i\in N_g} p_i d_i' \;>\; \sum_{i\in N_g} p_i d_i^\star.
\]
Let $i^\star$ be the index in \(N_g\) with the smallest $p_i$ for which $d_{i^\star}' > d_{i^\star}^\star$ (such an index exists because the weighted sum under $\dvec'$ is strictly larger). Again, by the preliminary reduction we may assume complement coordinates in $\dvec'$ are already adjusted to give minimal complement mass (maximal group mass) consistent with feasibility. The algorithm's `drop' construction initializes group coordinates at their upper bounds and complement coordinates at their lower bounds, then uses the complement-first, ascending-$p$ projection (which reduces the smallest-$p$ group members first) to realize the allocation that maximizes $\sum_{i\in N_g} p_i d_i$ (equivalently minimizes $S_g$). Therefore no feasible $\dvec'$ can achieve a strictly larger $\sum_{i\in N_g} p_i d_i$ than $\dvec^\star$, a contradiction.

Both exhaustive cases lead to contradictions. Hence no feasible $\dvec'\in\mathfrak{D}(\hat\dvec,\gamma)$ has strictly larger loss than $\dvec^\star$, and $\dvec^\star$ is a global maximizer of \(\mathcal L(\pvec,\cdot)\).
\end{proof}

\subsection{Methods for prediction of $\hatd$}\label{app:hatd}
We construct $\hatd$ with a simple independent action model, calculating dropout probabilities as a function of two features that are shared across instances: age and gender. 

Although these features are shared, the particular values are not consistent across instances. Thus, we preprocess the demographic panel data to ensure consistent feature representations across the multiple datasets. We standardize the \texttt{age} and \texttt{gender} features and learn dropout probabilities with these features and their second order co-occurrences. 

\paragraph{Age Standardization}

Age data were reported in across datasets, using various age ranges.

To standardize ages across datasets, we first converted reported age ranges into single numeric ages by sampling uniformly within the reported range. 

For an age range ``$a$--$b$'', we sampled an integer uniformly from $[a, b]$ (inclusive of $a$. For open-ended ranges like ``75+'', we sampled uniformly between $75$ and $85$ as a reasonable upper bound.  For single numeric ages, we used the reported value directly.

After assigning numeric ages, participants were bucketed into consistent age groups. 

\paragraph{Gender Standardization}

Gender entries were also partially heterogeneous across datasets. We standardized gender into four categories:

\begin{itemize}
    \item \textbf{Female}: entries such as ``female'', ``woman'', ``girl'', ``f''.
    \item \textbf{Male}: entries such as ``male'', ``man'', ``boy'', ``m''.
    \item \textbf{Nonbinary}: entries such as ``non-binary'', ``nonbinary''
    \item \textbf{Other}: all remaining or missing entries.
\end{itemize}

\paragraph{Learning Dropout Probabilities}

We model the probability that a selected participant drops out of a panel using multiplicative beta parameters associated with individual features and their interactions. 

\textbf{Feature tuples}: For a chosen feature list (e.g., \texttt{age}, \texttt{gender}), we generate all unique tuples of features and values to account for second-order interactions). This defines new features and values of the form $f', v' = (f^1-f^2, v^1-v^2)$. For each panel instance, we collect all unique value combinations for each feature tuple among selected participants, and we learn beta parameters on each of them. 

\textbf{Likelihood function}: For each participant, we compute the probability a person drops out as the complement of the probability they \emph{stay} in the panel as
    \[
    1 - \hatd_i = \Pr(\text{person $i$ stays in}) = \beta_0 \prod_{f \in F} \beta_{f,f(i)},
    \]
    where $\beta_0$ is a base probability, $F$ is the set of feature tuples, and $v = f(i)$ is the participant's value for feature or feature tuple $f$. These parameters are estimated via MLE, with the implementation matching that of \cite{assos2025alternates}.

We provide summary statistics of our estimated dropout probabilities for each instance $i$ below. For each instance $i$, we find the $\hatd$ by building a model with parameters derived from the features, values and dropouts of all instances \textit{except} instance $i$ . 
\begin{center}
    \begin{tabular}{lcccc}
\toprule
Instance & Min & Max & Mean & Std Dev \\
\midrule
Instance 1 & 0.0004 & 0.2778 & 0.1157 & 0.0856 \\
Instance 2 & 0.0003 & 0.3750 & 0.1433 & 0.1351 \\
Instance 3 & 0.0004 & 0.6769 & 0.1241 & 0.1116 \\
Instance 4 & 0.0004 & 0.3662 & 0.1874 & 0.1115 \\
Instance 5 & 0.0004 & 0.4616 & 0.2115 & 0.0959 \\
Instance 6 & 0.0004 & 0.4286 & 0.1334 & 0.1022 \\
\bottomrule
\end{tabular}
\end{center}

\subsection{Convergence Results} \label{app:convergence}
We provide plots demonstrating convergence to a stable loss value over rounds below. 
\begin{figure}[H]
    \centering
    
    \includegraphics[width=0.4\linewidth]{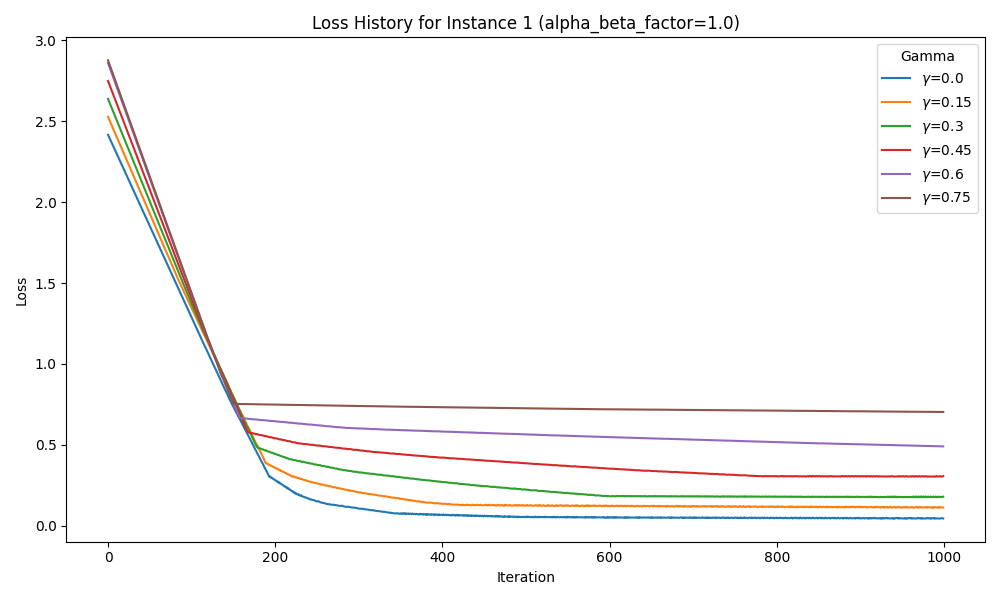}
    \includegraphics[width=0.4\linewidth]{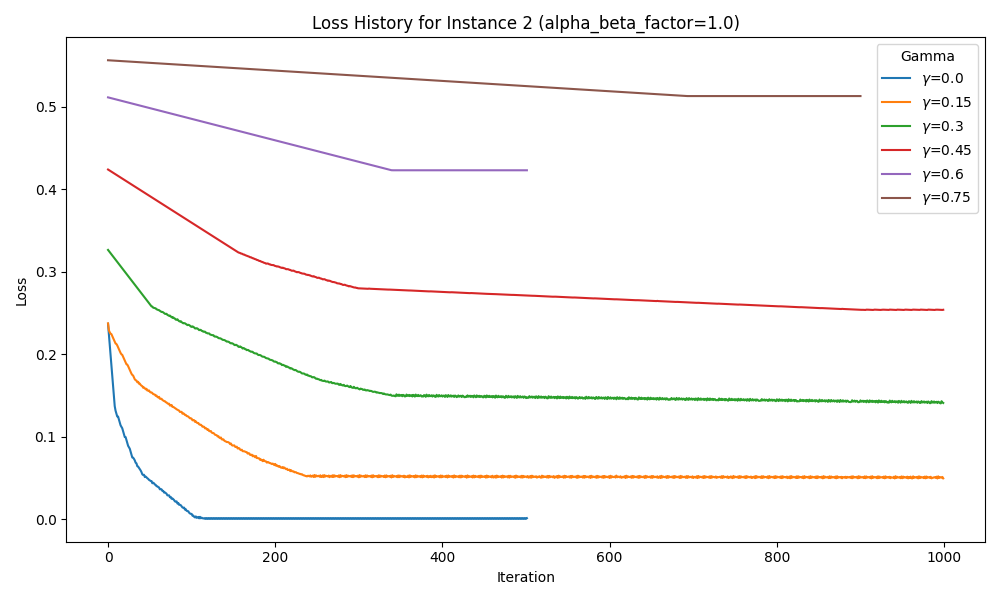}\\
    \includegraphics[width=0.4\linewidth]{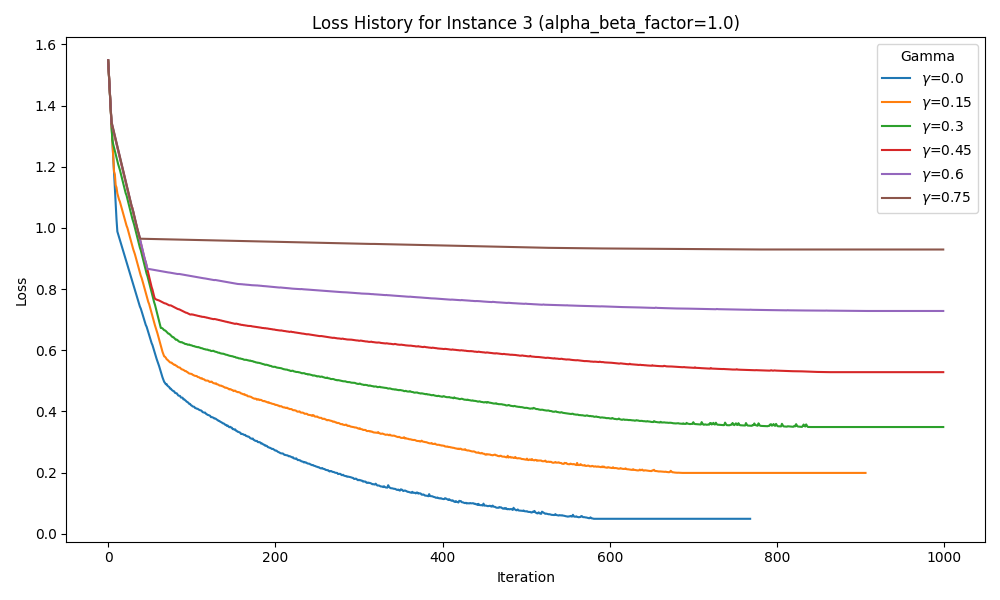}
    \includegraphics[width=0.4\linewidth]{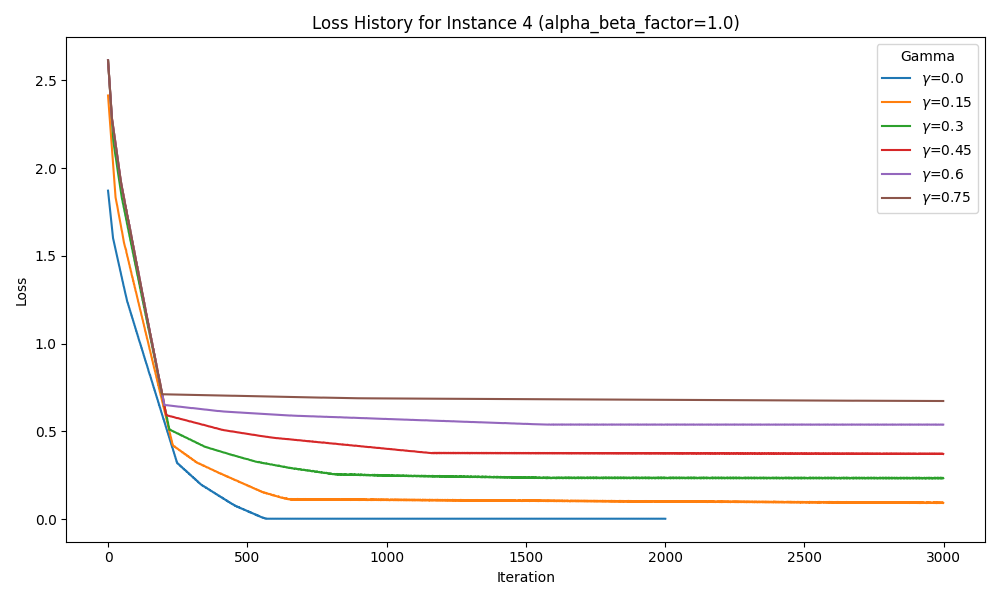}
    \includegraphics[width=0.4\linewidth]{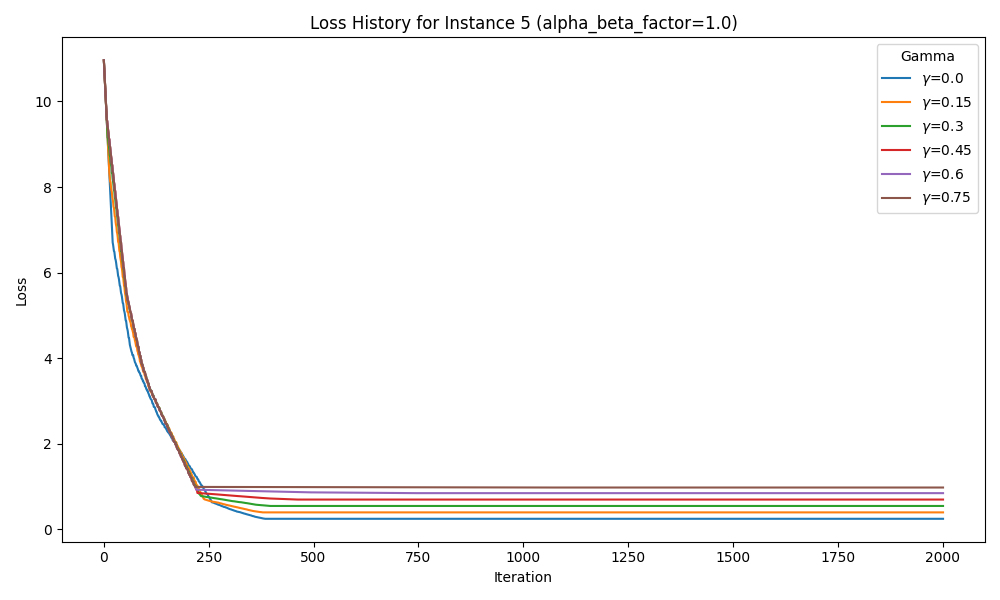}
    \includegraphics[width=0.4\linewidth]{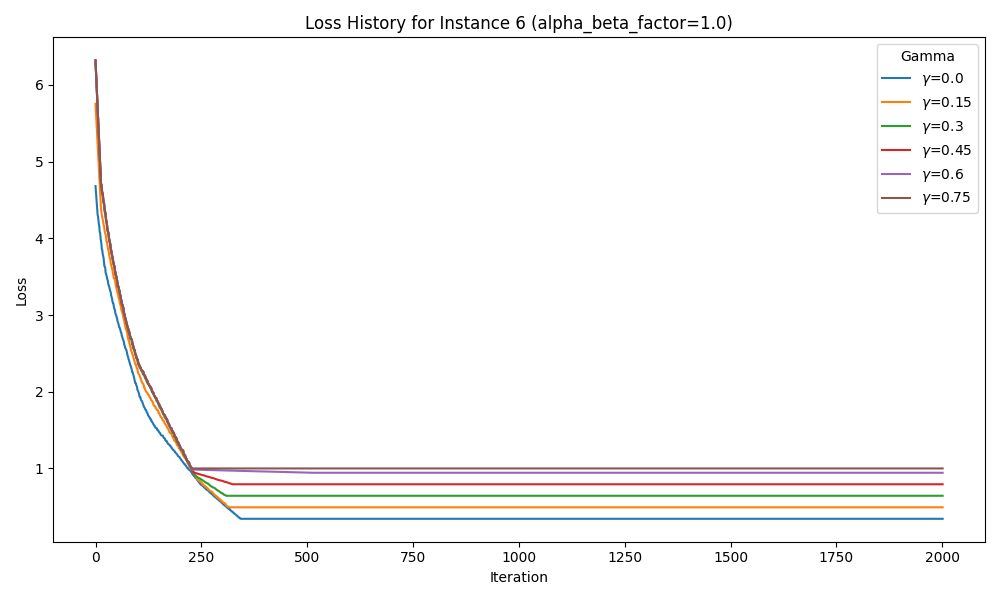}
\end{figure}

\subsection{Implementation of ERM Benchmark} \label{app:erm-benchmark}

We make two key alterations to the ERM algorithm presented in main text of \cite{assos2025alternates}: selecting an \textit{entire panel} instead of alternates, and optimizing on our loss function, rather than the loss function from the paper - for consistency. While the core contribution of the paper, as outlined in \cref{sec:intro}, selects a set alternates from which replacements can be drawn, this does not exactly match our work. Instead, we consider the extension of the algorithm in which ERM optimization is used to select a maximally robust panel from the start, parallel to our approach. Our ERM Benchmark implementation is from Appendix F.3 of \cite{assos2025alternates}. Assos et. al also measures loss as \textit{either} the \textit{sum} of all (scaled) deviation over all feature values, or a binary measure of whether all quotas were achieved. We instead optimize and evaluate on our loss function from \cref{section:model-new}, which measures the \textit{maximum} scaled deviation over any feature value. Both of these choices are in an effort to make the algorithms more comparable in both approach and evaluation.

\subsection{Remaining instances: Figure 1} 
\label{app:more_heatmaps}
\begin{figure}[H]
    \centering
    \includegraphics[width=0.4\linewidth]{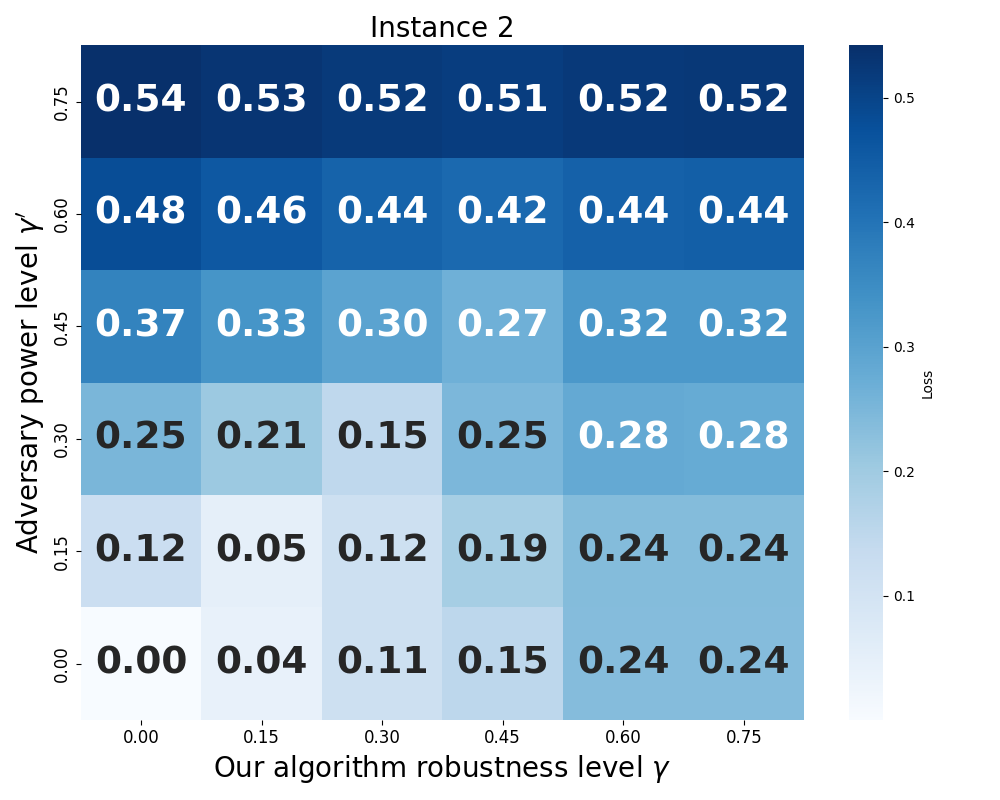}
    \includegraphics[width=0.4\linewidth]{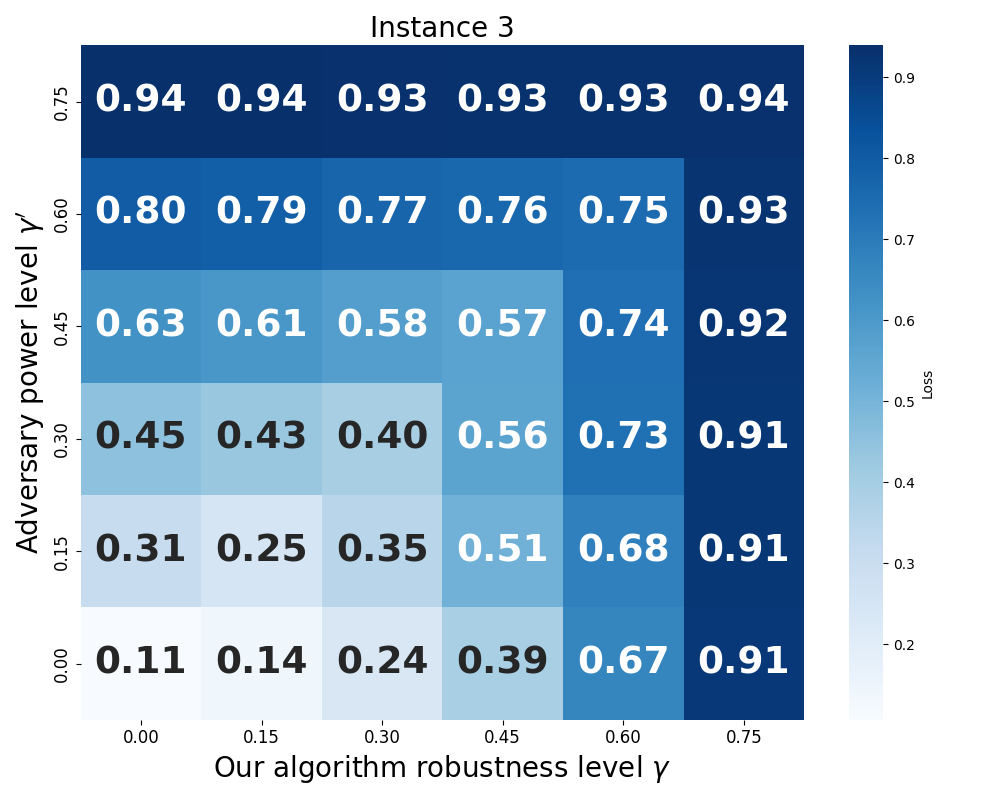}
    \includegraphics[width=0.4\linewidth]{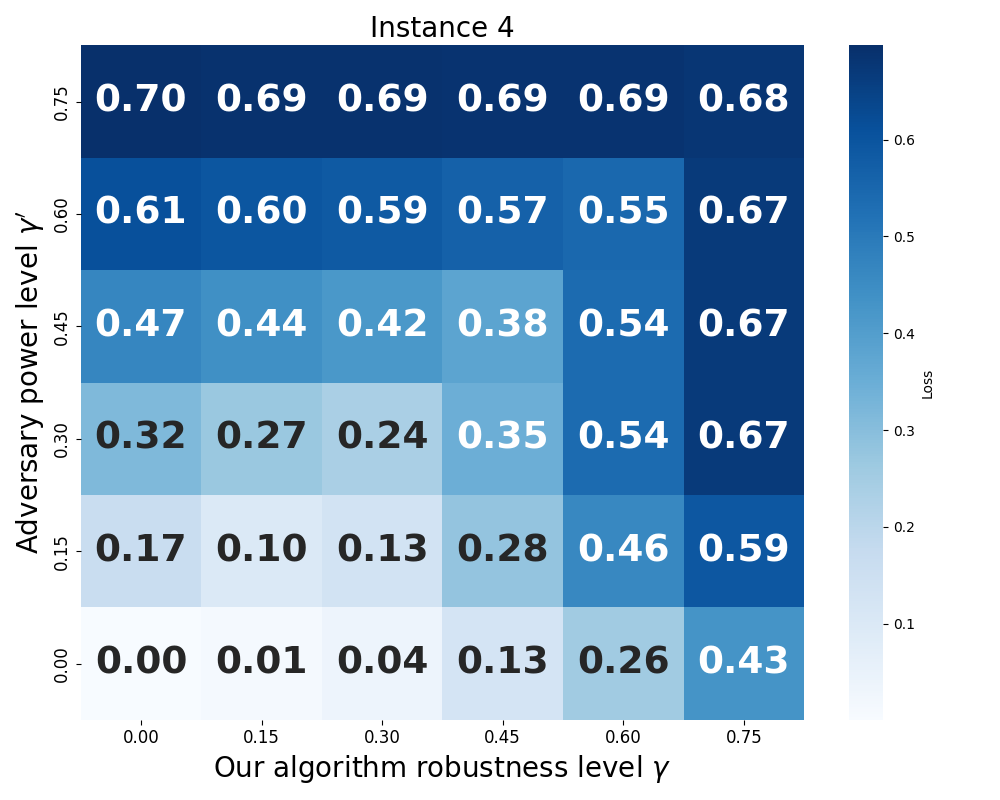}\\
    \includegraphics[width=0.4\linewidth]{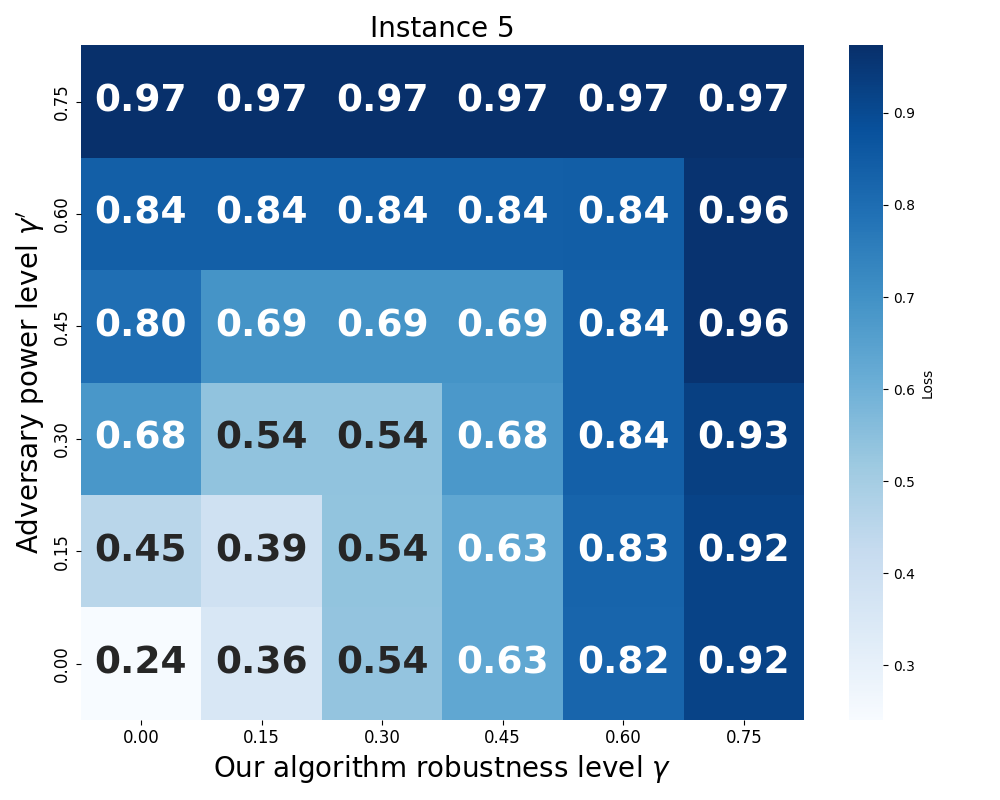}
    \includegraphics[width=0.4\linewidth]{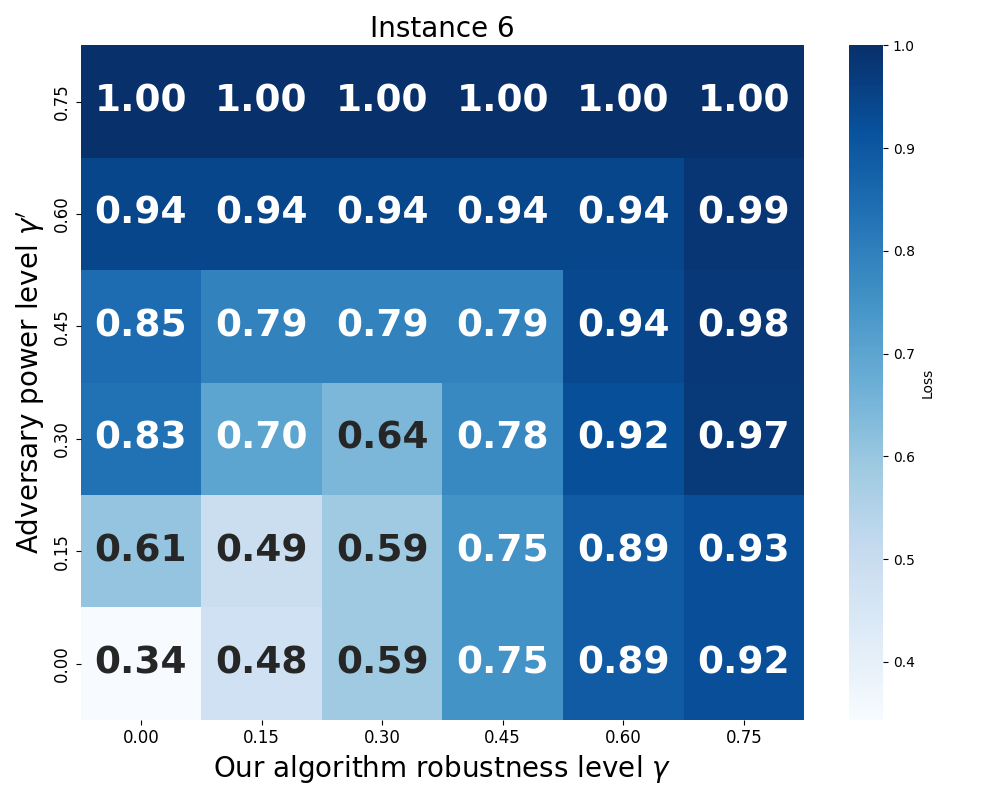}
\end{figure}

\subsection{Remaining instances: Figure 3} \label{app:more_benchmark}
\begin{figure}[H]
    \centering
    \includegraphics[width=0.4\linewidth]{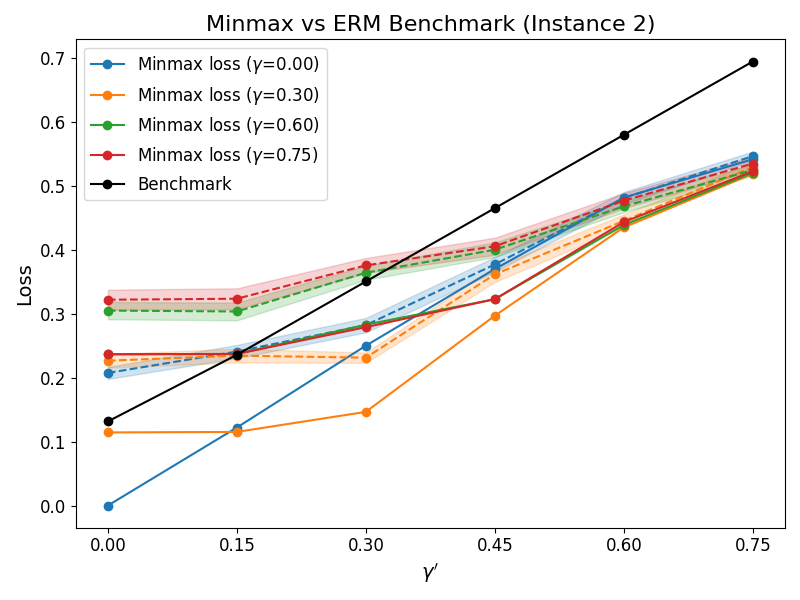}
    \includegraphics[width=0.4\linewidth]{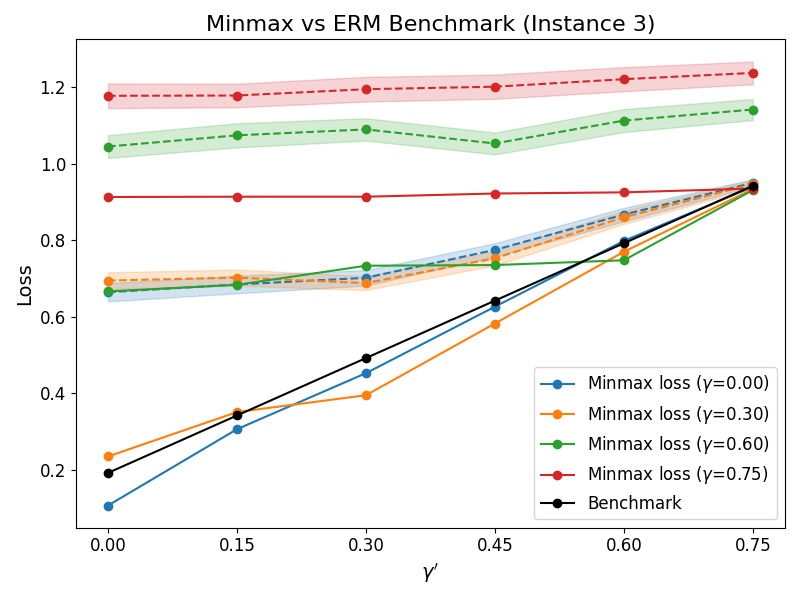}
    \includegraphics[width=0.4\linewidth]{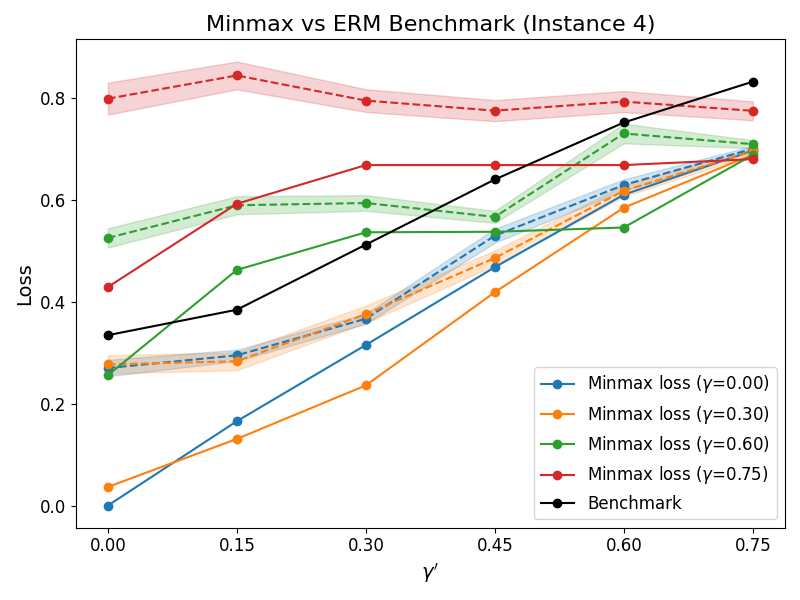}\\
    \includegraphics[width=0.4\linewidth]{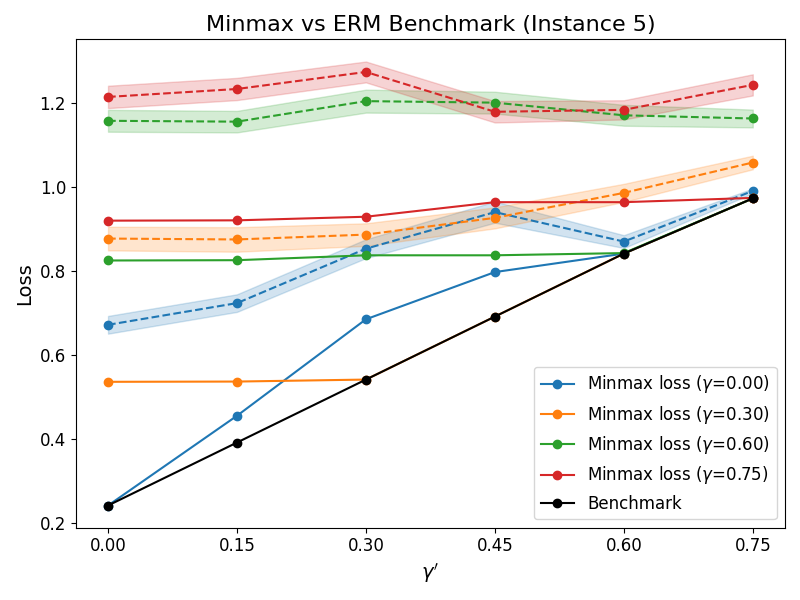}
    \includegraphics[width=0.4\linewidth]{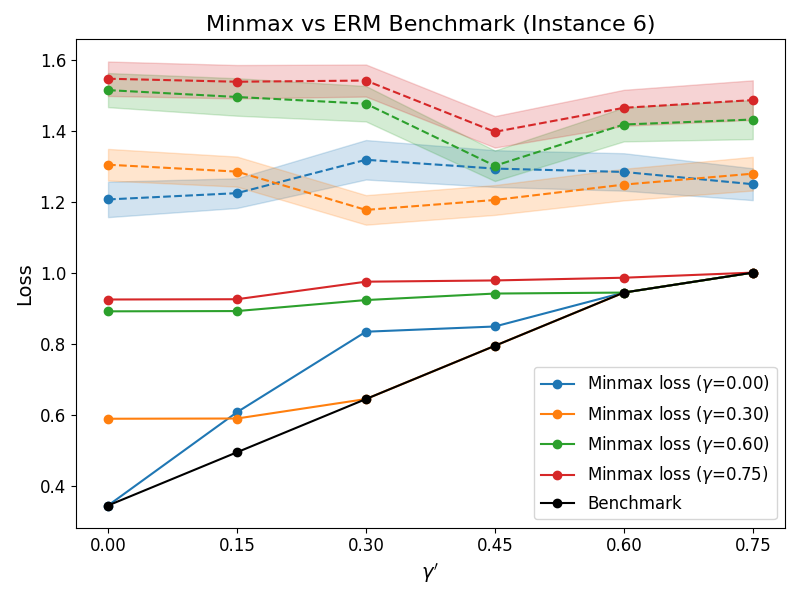}
\end{figure}

\subsection{Selection probabilities of \textsc{ERM-Benchmark}}\label{app:selprobs}
Although the ERM Benchmark returns a deterministic panel of people, it does randomize over identical pool members with equivalent feature value pairs. This makes intuitive sense, as it makes no difference to the objective if we replace pool member with one that is demographically identical. Thus, we calculate the selection probability of each panelist as: 

$$\Pr[\text{Pool member $i$ is selected for the panel}] = \frac{{\sum_{j = 1}^n \mathbf{1}[f(j) = f(i) \quad \forall f \in F ]}}{\text{n}}$$ 

i.e. the proportion of the pool that is people with identical feature-value composition to person $i$.

\end{document}